\documentclass[11pt,onecolumn]{IEEEtran}
\usepackage[utf8]{inputenc}
\usepackage[T1]{fontenc}
\usepackage{amsmath,amssymb,amsthm,bm}
\usepackage[noadjust]{cite}
\usepackage{booktabs}
\usepackage{float}
\usepackage{needspace}
\usepackage[mathscr]{eucal}
\usepackage[colorlinks=true,allcolors=blue]{hyperref}
\usepackage[capitalise,noabbrev,nameinlink]{cleveref}

\newtheorem{theorem}{Theorem}
\newtheorem{lemma}{Lemma}
\newtheorem{definition}{Definition}
\newtheorem{remark}{Remark}
\newtheorem{proposition}{Proposition}
\newtheorem{corollary}{Corollary}
\newtheorem{conjecture}{Conjecture}

\newcommand{\F}{\mathbb{F}}
\newcommand{\Tr}{\operatorname{Tr}}
\newcommand{\Span}{\operatorname{Span}}
\newcommand{\lcmop}{\operatorname{lcm}}
\newcommand{\RS}{\operatorname{RS}}
\newcommand{\crf}{\mathscr R}
\newcommand{\Pone}{\mathbb P^1}
\newcommand{\set}[1]{\left\{#1\right\}}
\newcommand{\mcff}{m_{\mathrm{cff}}}
\newcommand{\ellmin}{\ell_{\min}}
\newcommand{\primorial}[1]{p_{#1}^{\#}}
\begin{document}

\title{The Minimum Subpacketization of Reed-Solomon Codes with Linear Repair Achieving the Cut-Set Bound}

\author{%
Weijun Fang, Jing Qiu, Shu-Tao Xia, and Fang-Wei Fu%
\thanks{This work was supported in part by the National Key Research and Development Program of China under Grants 2022YFA1004900 and 2022YFA1005000; by the National Natural Science Foundation of China under Grants 62571298, 62571301, and 62371259;
by the Fundamental Research Funds for the Central Universities of China (Nankai University); by the Nankai Zhide Foundation; and by the China Postdoctoral Science Foundation under Grant 2025M773088. (Corresponding author: Jing Qiu.)}
\thanks{Weijun Fang is with the State Key Laboratory of Cryptography and Digital Economy Security, the Key Laboratory of Cryptologic Technology and Information Security, Ministry of Education, and the School of Cyber Science and Technology, Shandong University, Qingdao 266237, China (e-mail: fwj@sdu.edu.cn).}
\thanks{Jing Qiu and Shu-Tao Xia are with the Tsinghua Shenzhen International Graduate School, Tsinghua University, Shenzhen, Guangdong 518055, China 
(e-mails: \{jing\_qiu, xiast\}@sz.tsinghua.edu.cn).}
\thanks{Fang-Wei Fu is with the Chern Institute of Mathematics and LPMC, Nankai University, Tianjin 300071, China (e-mail: fwfu@nankai.edu.cn).}%
}
\maketitle

\begin{abstract}
In this paper, we study the minimum subpacketization of $(n,k)$
Reed-Solomon (RS) codes over $\mathbb{F}_{q^\ell}$
admitting cut-set-optimal exact $\mathbb{F}_q$-linear
repair of every coordinate from any $d$ helpers.
For $3\le k<d\le n-1$, write $s=d+1-k$ and
denote this minimum by $\ellmin(n,k,d;q)$.
Using divisibility and relative-degree constraints
on cross-ratio fields of evaluation sets, we
establish $\ellmin(n,k,d;q)\ge sG_s(d+1)$, where
$G_s(m)$ is the minimum product of $m$ pairwise
coprime integers, each at least $s$.
Let $\Delta_s(n,d)$ be the minimum of
$\lcmop(a_1,\ldots,a_n)$ over positive integers
$a_1,\ldots,a_n$ satisfying
$\lcmop(a_j:j\in J\cup\{i\})\ge
s\lcmop(a_j:j\in J)$ for every $d$-subset
$J\subseteq[n]$ and every $i\in[n]\setminus J$.
We construct RS codes with subpacketization
$s\Delta_s(n,d)$ and further prove that
$\ellmin(n,k,d;q)=s\Delta_s(n,d)=sG_s(n)$ whenever
$d(d+2)\geq 4(n+1)$, for every prime power $q$.
In particular, this determines the minimum subpacketization in the full-helper case $d=n-1$, resolving the corresponding exact linear-repair case of Open Problem 20 posed by Ramkumar et al. in their survey
on codes for distributed storage. 
We further give a finite formula for $G_s(n)$
and determine the exact threshold for its
closed-form expression in terms of the first
$n$ primes.
For fixed $k$ and $d$, we prove
$\log\ellmin(n,k,d;q)=
\Theta_{k,d}(\log n\log\log n)$ as
$n\to\infty$, using divisor-function estimates
and cover-free families.
An affine-plane construction shows that
$\Delta_s(n,d)<G_s(n)$ for some parameters,
so $sG_s(n)$ is not a universal formula.
We conjecture that
$\ellmin(n,k,d;q)=s\Delta_s(n,d)$
for every prime power $q$ and all
$3\le k<d\le n-1$.
\end{abstract}

\begin{IEEEkeywords}
Reed-Solomon codes, linear repair, cut-set bound,
subpacketization, cross-ratios.
\end{IEEEkeywords}

\section{Introduction}

Scalar maximum-distance-separable (MDS) codes are a basic model for
storage-optimal erasure protection.  An $(n,k)$ scalar MDS code over a
field $K$ stores one $K$-symbol in each of $n$ nodes and allows the
original $k$ information symbols to be recovered from any $k$ nodes.
Equivalently, it tolerates any $n-k$ node erasures with the minimum
possible redundancy.  This global recovery property, however, does not
ensure communication-efficient repair of an individual failed node.
A direct repair procedure downloads one complete $K$-symbol from each
of $k$ surviving nodes and reconstructs the entire codeword, even
though only one symbol is missing.

This inefficiency motivates the array-code viewpoint.  In an MDS array
code over a base field $F$, each node stores a vector of $\ell$
subsymbols, and a helper may transmit functions of its stored vector
rather than its entire contents.  The regenerating-code framework of
Dimakis \emph{et al.} \cite{Dimakis2010} formalizes the resulting
tradeoff between storage and repair communication.  If a failed node
is repaired from $d$ helpers and $\gamma$ denotes the total number of
downloaded $F$-symbols, then every MDS array code satisfies the cut-set
bound
\begin{equation}
	\gamma\ge \frac{d\ell}{d+1-k}.
	\label{eq:cutset}
\end{equation}
An MDS array code attaining equality in
\eqref{eq:cutset} is a minimum-storage regenerating (MSR) code.  Much of the literature on
exact-repair MSR codes has focused on MDS array codes; representative
constructions include the product-matrix construction
\cite{RashmiShahKumar2011} and high-rate constructions with optimal
repair bandwidth \cite{YeBarg2017}.  The node size $\ell$, called the
\emph{subpacketization}, determines the granularity available to a
repair scheme and is a principal measure of its algebraic and
implementation complexity.  Fundamental lower bounds on the
subpacketization of MSR array codes were developed in
\cite{GoparajuTamoCalderbank2014,AlrabiahGuruswami2021,
	BalajiVajhaKumar2022}.  On the construction side, Liu, Li, and Tang
constructed high-rate optimal-access MSR codes over small finite fields
\cite{LiuLiTang2018}, while Li, Tang, and Tian developed a generic
transformation that, when applied iteratively, enables optimal repair
of all nodes with controlled subpacketization \cite{LiTangTian2018}.  Liu, Li, and Tang later extended
this line to MDS array codes supporting multiple repair degrees, with
smaller subpacketization than earlier multi-degree constructions
\cite{LiuLiTang2023}.  More recently, Wang, Li, Hu, and Ye constructed
MSR codes with $d=k+1$ and subpacketization $2^{\lceil n/3\rceil}$ while maintaining a
field size linear in $n$ \cite{WangLiHuYe2023}.  With $s=d+1-k$, Li, Wang, Hu, and Ye
subsequently gave, for every admissible repair degree, optimal-access
MSR codes over fields of size $O_s(n)$ with subpacketization
$\ell=s^{\lceil n/s\rceil}$
and also non-optimal-access MSR codes with the smaller subpacketization
$s^{\lceil n/(s+1)\rceil}$ \cite{LiWangHuYe2024}.
These developments highlight the central role of subpacketization in
general MSR array-code design and motivate the corresponding minimum
subpacketization problem for the much more rigid class of scalar
Reed-Solomon codes.

Reed-Solomon (RS) codes \cite{ReedSolomon1960} are one of the most important families of scalar MDS codes and are widely used in communication and storage systems. Let $F=\F_q$ and $K=\F_{q^\ell}$. Although an RS code is a scalar code over $K$, after fixing an $F$-basis of $K$, each symbol in $K$ can be represented as a vector in $F^\ell$. Thus, the same code can be viewed as an MDS array code over $F$ with subpacketization $\ell$. This scalar-to-array viewpoint has also been used more generally in the repair of scalar MDS codes \cite{ShanmugamEtAl2014}.

This viewpoint allows the repair of RS codes to be studied within the regenerating-code framework. However, unlike unrestricted MDS array codes, RS codes are subject to additional algebraic constraints arising from their scalar extension field and polynomial-evaluation structure. Consequently, the repair bandwidth and subpacketization of RS codes have attracted substantial attention. In this paper, we require exact $F$-linear repair of every coordinate from every set of $d$ helper nodes with bandwidth meeting the cut-set bound in \eqref{eq:cutset}, and seek the minimum possible extension degree $\ell=[K:F]$.

\subsection{Linear Repair of Reed-Solomon Codes}
The systematic study of bandwidth-efficient linear repair of scalar RS
codes was initiated by Guruswami and Wootters \cite{GW2017}. They
established a dual-code characterization of exact $F$-linear repair,
reducing the design of a repair scheme to the construction of suitable
collections of dual codewords, and used trace maps to obtain
nontrivial repair schemes for full-length RS codes. Building on this
framework, Dau and Milenkovic \cite{DauMilenkovic2017}, and later Dau
\emph{et al.} \cite{DauEtAl2021}, developed optimal repair schemes for
additional families of full-length RS codes, including constructions
based on subspace polynomials. Here optimality refers to their
full-length families and specified redundancy ranges under their
linear-repair lower bounds; it need not mean equality in
\eqref{eq:cutset}. Li, Wang, and Jafarkhani
\cite{LiWangJafarkhani2019} further investigated explicit tradeoffs
between repair bandwidth and subpacketization.

A major breakthrough was achieved by Tamo, Ye, and Barg
\cite{TYB2017,TYB2019}, who constructed RS codes attaining the cut-set
bound for the repair of any single failed node from any prescribed set
of $d$ helper nodes. With $s=d+1-k$, the subpacketization of their construction is
\[
\ell=s\prod_{i=1}^{n}\tilde p_i,
\]
where $\tilde p_1,\ldots,\tilde p_n$ are distinct primes satisfying
$\tilde p_i\equiv1\pmod s$. Their framework also extends to multiple
erasures. Subsequent work developed optimal-access and
error-resilient repair schemes \cite{ChenYeBarg2020}, while
near-optimal repair bandwidth with substantially smaller
subpacketization was studied in \cite{GuruswamiJiang2019}.

In our prior work \cite{Qiu2026}, we developed a
basis-transformation method for constructing the required repair
subspaces for every integer relative extension degree $p\ge s$,
without the congruence restriction $p\equiv1\pmod s$.
Applying this method to the Tamo-Ye-Barg framework, we obtained
RS codes attaining the cut-set bound using arbitrary distinct prime
degrees $p_i\ge s$.

These constructions establish the existence of cut-set-optimal RS
repair schemes, but do not determine the least possible subpacketization:
assigning distinct prime degrees is only one way to meet the repair
requirements.  The present paper studies the minimum over all RS
evaluation sets, allowing more general degree assignments and deriving
lower bounds that apply to every such code.

\subsection{Main Results}
For a prime power $q$ and integers $3\le k<d\le n-1$, let
$\ellmin(n,k,d;q)$ denote the smallest integer $\ell$ for which
an $(n,k)$ RS code over $\mathbb F_{q^\ell}$ admits cut-set-optimal
exact $\F_q$-linear repair of every coordinate from every
$d$-helper set.  Write
$s=d+1-k\ge2$.

Our results are expressed in terms of two arithmetic quantities.
For $m\ge1$, define
\[
G_s(m)=\min\left\{\prod_{i=1}^m g_i:
  g_i\in\mathbb Z,\ g_i\ge s,\
  \gcd(g_i,g_j)=1\text{ for }i\ne j\right\},
\]
with $G_s(0)=1$.  For $1\le t<N$, let $\Delta_s(N,t)$ be the
minimum of $\lcmop(a_1,\ldots,a_N)$ over positive integers
$a_1,\ldots,a_N$ satisfying
\begin{equation}
\lcmop(a_j:j\in J\cup\{i\})
\ge s\,\lcmop(a_j:j\in J)
\label{eq:intro-Delta-feasible}
\end{equation}
for every $i\in[N]$ and every $t$-subset
$J\subseteq[N]\setminus\{i\}$.
The first quantity minimizes a product of pairwise coprime degrees;
the second permits shared factors, subject to the relative-degree
requirement imposed by every failed-coordinate/helper-set pair.
In particular, $\Delta_s(n,d)\le G_s(n)$.

Our principal result determines the minimum subpacketization exactly.
For every prime power $q$ and all $3\le k<d\le n-1$ satisfying
$d(d+2)\geq 4(n+1)$, or equivalently
$d\ge\sqrt{4n+5}-1$, we prove
\begin{equation}
\ellmin(n,k,d;q)=s\Delta_s(n,d)=sG_s(n).
\label{eq:intro-main-exact}
\end{equation}
Thus the exact formula applies even when the number of helpers is
of order $\sqrt n$.
For full-helper repair ($d=n-1$), this yields
\[
\ellmin(n,k,n-1;q)=(n-k)G_{n-k}(n),
\qquad 3\le k<n-1.
\]
This settles the exact $\F_q$-linear repair case, in the stated
dimension range, of Open Problem~20 posed by Ramkumar \emph{et al.}
\cite[Open Problem~20]{RamkumarEtAl2022}.

The main ingredients of the proof and our further contributions are as follows.
\begin{enumerate}
\item Our lower bounds come from the divisibility and relative-degree
constraints imposed by cut-set-optimal repair on the fields generated
by projective cross-ratios of the evaluation points.  Applying these
constraints to a punctured $(d+1,k)$ RS code in which every
coordinate admits cut-set-optimal full-helper repair yields
\[
\ellmin(n,k,d;q)\ge sG_s(d+1),
\]
which strengthens the lower bound of Tamo, Ye, and Barg
\cite[Theorem~5]{TYB2017} for scalar MDS codes when specialized to RS codes.
By using the constraints from all failed-coordinate/helper-set pairs,
we first establish
$\ellmin(n,k,d;q)\ge s\Delta_s(n,d)$ for $2d\geq n+3$.
To go beyond this range, we combine compositum identities for
cross-ratio fields with a finer analysis of the private and shared
prime factors of their degrees.  Consequently, we establish
$\ellmin(n,k,d;q)\ge sG_s(n)$ whenever
$d(d+2)\geq 4(n+1)$.

\item For the matching upper bound, we extend the distinct-prime
constructions to general feasible degree assignments.
Our prior work \cite[Theorem~3]{Qiu2026} provides the necessary
repair-subspace decomposition for every integer relative degree
$p\ge s$.  We show that every tuple satisfying
\eqref{eq:intro-Delta-feasible} with $(N,t)=(n,d)$ yields an RS code
with subpacketization $s\lcmop(a_1,\ldots,a_n)$ that achieves the
cut-set bound for every repair configuration.  Unlike the earlier
constructions with pairwise distinct prime degrees, these tuples
may contain composite degrees with shared prime factors.  Minimizing over all feasible tuples, we obtain
\[
\ellmin(n,k,d;q)\le s\Delta_s(n,d)
\]
for every admissible parameter set and every prime power $q$.
Together with the lower bounds, this proves \eqref{eq:intro-main-exact}.

\Needspace{8\baselineskip}
\item We also study the arithmetic quantity $G_s(m)$ appearing in the
exact formula. In addition to a finite formula valid for all $m$,
we determine the exact threshold for the closed-form expression given below. Let $p_j$ denote the $j$-th prime, and write
$\primorial{m}=\prod_{j=1}^{m}p_j$ for $m\ge1$, with
$\primorial{0}=1$. Also put $A_s=\lcmop(1,\ldots,s-1)$.  For $s\ge3$, let $u_s$
and $v_s$ be the largest prime and the largest prime power,
respectively, strictly below $s$, and set $u_2=v_2=1$.  Then
\[
G_s(m)=A_sp_m^{\#}
\quad\Longleftrightarrow\quad
m\ge\pi(u_sv_s),
\]
where $\pi(x)$ counts the primes at most $x$.  In particular,
$G_2(m)=p_m^{\#}$ for every $m$, giving an explicit expression for
\eqref{eq:intro-main-exact} when $s=2$.

\item For fixed $k$ and $d$ (and any fixed prime power $q$), we further prove
\[
\log\ellmin(n,k,d;q)
=\Theta_{k,d}(\log n\log\log n)
\qquad(n\to\infty).
\]
The lower bound has leading coefficient $(d-3)/\log2$, and the
matching-order upper bound follows from cover-free families and
our degree construction. By comparison, for fixed $s$, the earlier constructions using the smallest admissible distinct prime degrees satisfy $\log\ell=\Theta(n\log n)$.

\end{enumerate}

Finally, the equality $\ellmin(n,k,d;q)=sG_s(n)$ does not hold
for all admissible parameters.  An affine-plane construction in
Appendix~\ref{app:affine-obstruction} exhibits parameters for which
$\Delta_s(n,d)<G_s(n)$ and hence
$\ellmin(n,k,d;q)<sG_s(n)$.  This distinction suggests that
$\Delta_s(n,d)$, rather than $G_s(n)$, governs the general problem.
We conjecture that
\[
\ellmin(n,k,d;q)=s\Delta_s(n,d)
\]
for every prime power $q$ and all $3\le k<d\le n-1$.

\subsection{Organization}
Section~\ref{sec:prelim} reviews linear repair and the
basis-transformation result used in our construction.
Section~\ref{sec:lower-bounds} develops the cross-ratio constraints
and proves the general and parameter-dependent lower bounds.
Section~\ref{sec:construction} constructs RS codes from feasible
degree tuples, establishes \eqref{eq:intro-main-exact}, and derives
exact formulas for $G_s(m)$.  Section~\ref{sec:fixed-parameters} determines the
fixed-parameter asymptotic order of the minimum subpacketization.
The appendices give auxiliary prime-product estimates and an
affine-plane construction exhibiting $\Delta_s(n,d)<G_s(n)$.

\section{Preliminaries}
\label{sec:prelim}

We write $\mathbb Z$ for the integers and $[m]=\{1,\ldots,m\}$
for any positive integer $m$. All logarithms without an
explicitly indicated base are natural. Table~\ref{tab:notation} summarizes
the principal notation; symbols specific to individual arguments
are defined where needed.

\begin{table}[H]
\centering
\caption{Principal notation.}
\label{tab:notation}
\begingroup
\small
\renewcommand{\arraystretch}{1.15}
\setlength{\tabcolsep}{8pt}
\newsavebox{\RSnotationbox}
\sbox{\RSnotationbox}{%
\begin{tabular}{@{}ll@{}}
\toprule
\textbf{Notation} & \textbf{Description} \\
\midrule
$n,k,d$ & Code length, dimension, and number of helpers. \\
$s$ & Repair parameter $d+1-k$. \\
$F$ & Base field $\mathbb F_q$ for linear repair. \\
$K$ & Code-symbol field $\mathbb F_{q^\ell}$. \\
$\ell$ & Subpacketization $[K:F]$. \\
$C$ & $(n,k)$ RS code over $K$. \\
$\Omega$ & Set of $n$ distinct evaluation points in $K$. \\
$\ellmin(n,k,d;q)$ & Minimum $\ell$ for cut-set-optimal $F$-linear repair of every coordinate from every $d$-helper set. \\
$p_j$, $\primorial{m}$ & $j$-th prime; $\primorial{m}=\prod_{j=1}^m p_j$ ($\primorial{0}=1$). \\
$\crf_F(X)$ & Field generated over $F$ by the cross-ratios of $X$. \\
$G_s(m)$ & Minimum product of $m$ pairwise coprime integers $\ge s$. \\
$\Delta_s(N,t)$ & Minimum LCM over feasible $(N,t,s)$-tuples. \\
\bottomrule
\end{tabular}%
}
\ifdim\wd\RSnotationbox>\linewidth
  \PackageError{rsnotation}{Notation table width \the\wd\RSnotationbox\space exceeds \the\linewidth}{Shorten the cell text; do not permit an overflowing table.}%
\fi
\usebox{\RSnotationbox}
\endgroup
\end{table}

\subsection{Reed-Solomon Codes and Linear Repair}
Throughout, let $F=\F_q\subseteq K=\F_{q^\ell}$ and
let $\alpha_1,\ldots,\alpha_n$ be pairwise distinct elements of $K$, and
put $\Omega=\{\alpha_1,\ldots,\alpha_n\}$. The Reed-Solomon code with evaluation set
$\Omega$ is defined as
\[ C=\RS_K(n,k,\Omega)
 =\set{(f(\alpha_1),\ldots,f(\alpha_n)):
       f\in K[x],\ \deg f<k}.\]
Its Euclidean dual, with respect to the inner product
$\sum_{j=1}^n x_jy_j$, is the generalized RS code
\[ C^\perp=
 \set{(w_1g(\alpha_1),\ldots,w_ng(\alpha_n)):
       g\in K[x],\ \deg g<n-k},\]
where $w_j=\prod_{t\ne j}(\alpha_j-\alpha_t)^{-1}$.

Viewing $K$ as an $\ell$-dimensional vector space over $F$ turns each
code coordinate into a node containing $\ell$ $F$-symbols. An
$F$-linear repair scheme for a failed coordinate $i$ downloads
$F$-linear functionals of the symbols stored in a helper set
$H\subseteq[n]\setminus\{i\}$ and reconstructs the missing element of
$K$. For $|H|=d$, put
\[
s=d+1-k.
\]
The cut-set bound requires any exact repair scheme to download at
least $d\ell/s$ $F$-symbols in total. We call a repair scheme
cut-set optimal if equality holds. Throughout, we require
cut-set-optimal $F$-linear repair of every coordinate from every
$d$-helper set.

We use the following dual-code characterization of linear repair
from \cite[Theorem~4]{GW2017}.

\begin{theorem}[{\cite[Theorem~4]{GW2017}}]
\label{thm:GW}
Let $C\subseteq K^n$ be a scalar linear MDS code with $[K:F]=\ell$,
and fix $i\in[n]$. Then coordinate $i$ admits exact $F$-linear repair
from the other $n-1$ coordinates with bandwidth at most $b$ if and
only if there exist $\ell$ codewords
$\bm p^{(1)},\ldots,\bm p^{(\ell)}\in C^\perp$ such that
\[
\dim_F\Span_F\{p_i^{(t)}:1\le t\le\ell\}
=
\ell
\]
and
\[
\sum_{j\ne i}
\dim_F\Span_F\{p_j^{(t)}:1\le t\le\ell\}
\le b.
\]
\end{theorem}

In particular, in the full-helper case $d=n-1$, if the repair is
cut-set optimal and $s=n-k$, then the codewords in
\cref{thm:GW} satisfy
\[
\Span_F\{p_i^{(t)}:1\le t\le\ell\}=K
\]
and
\[
\sum_{j\ne i}
\dim_F\Span_F\{p_j^{(t)}:1\le t\le\ell\}
=
\frac{(n-1)\ell}{s}.
\]

\begin{remark}A simple but useful consequence is the following. If coordinate $i$
admits cut-set-optimal repair from a $d$-helper set $H$,
then puncturing all coordinates outside $H\cup\{i\}$ yields a length-$(d+1)$ MDS
code with cut-set-optimal full-helper repair of the same coordinate.
The parameter $s=d+1-k$ remains unchanged.  This reduction allows us
to focus on full-helper repair in many of the subsequent lower-bound
arguments.
\end{remark}

As in the Introduction, $\ellmin(n,k,d;q)$ denotes the
minimum subpacketization for cut-set-optimal $F$-linear repair
of every coordinate from every $d$-helper set. Coordinate permutations
and nonzero coordinate scalings preserve $F$-linear repair bandwidth.
Consequently, a generalized RS code and the ordinary RS code obtained
by removing its coordinate multipliers are equivalent for the present
purpose.

\subsection{Finite-Field Degrees and Local Repair Subspaces}

Let $\overline F$ be an algebraic closure of $F$.
For $x_1,\ldots,x_t\in\overline F$ and every nonempty
$I\subseteq[t]$, the finite-field compositum formula
\cite[Chapter~2]{LidlNiederreiter1997} gives
\[
[F(x_i:i\in I):F]
=\lcmop([F(x_i):F]:i\in I).
\]
In particular,
\[
\F_{q^a}\F_{q^b}=\F_{q^{\lcmop(a,b)}}.
\]

For a finite extension $K/L$ and $x\in K$, write
\[
\Tr_{K/L}(x)=\sum_{j=0}^{[K:L]-1}x^{|L|^j}.
\]
Finite-field extensions are separable, so the pairing
$K\times K\longrightarrow L$, $(x,y)\mapsto\Tr_{K/L}(xy)$, is nondegenerate
\cite[Chapter~2]{LidlNiederreiter1997}.

In the construction of Tamo, Ye, and Barg \cite{TYB2019}, the
evaluation points are assigned distinct prime degrees $\tilde p_i$ satisfying
\[
 \tilde p_i\equiv1\pmod{s}.
\]
This congruence condition is used in constructing the local
repair-enabling subspaces and restricts the available degree choices.

Our prior work \cite[Theorem~3]{Qiu2026} removed this restriction by a
basis-transformation argument and established the following subspace
decomposition for every integer relative degree $r\ge s$.
\begin{theorem}[{\cite[Theorem~3]{Qiu2026}}]
\label{thm:Qiu-local}
Let $F_0\subset E_0=F_0(\alpha)\subset K_0$ be finite fields
with
\[
[E_0:F_0]=r\ge s=[K_0:E_0].
\]
Then there exists an explicitly constructible $F_0$-subspace
$S_0\subseteq K_0$ satisfying
\[
\dim_{F_0}S_0=r,
\qquad
K_0=\bigoplus_{u=0}^{s-1}\alpha^uS_0.
\]
\end{theorem}

In Section~\ref{sec:construction}, this result will be applied
with the relative extension degree of a failed evaluation point over
the field generated by the chosen helpers.

\section{Lower Bounds on Subpacketization}
\label{sec:lower-bounds}
In this section, we establish lower bounds on the subpacketization. We first translate cut-set-optimal repair into degree constraints
on cross-ratio fields. Applying these constraints to a punctured
$(d+1,k)$ RS code yields a lower bound valid for all admissible
parameters.  We then combine the constraints from all helper sets.  For
$2d\geq n+3$, we obtain the lower bound $s\Delta_s(n,d)$.  Finally, a
more detailed analysis of the same degree system yields the stronger
bound $sG_s(n)$ whenever $d(d+2)\geq 4(n+1)$.

\subsection{Cross-Ratio Constraints}

We begin with the full-helper case, which will later be applied to a
failed coordinate together with a chosen helper set.  Throughout this
subsection, assume
\[
d=n-1,\qquad k\ge3,\qquad s=n-k\ge2,
\]
and write
\[
C=\RS_K(n,k,\Omega),
\qquad
\Omega=\{\alpha_1,\ldots,\alpha_n\}.
\]
Fix $i\in[n]$ and consider cut-set-optimal
$F$-linear repair of coordinate $i$.
We write $K[x]_{<s}$ for the space of
polynomials over $K$ of degree less than $s$, including the zero polynomial.
The following lemma records the dimension properties of the evaluation
spaces that follow from equality in the cut-set bound.  The underlying
argument is essentially the one used by Tamo, Ye, and Barg
\cite[proof of Theorem~5]{TYB2017}; we state it in the polynomial form
needed below.  See also \cite{TYB2019}.

\begin{lemma}\label{lem:evaluation-rigidity}
Suppose that the failed coordinate $i\in[n]$, with evaluation point
$\alpha_i$, admits cut-set-optimal $F$-linear repair from the other
$n-1$ coordinates. Then $s\mid\ell$, and there exists an
$\ell$-dimensional $F$-subspace $\mathcal P_i\subseteq K[x]_{<s}$
whose evaluation spaces, defined for
$y\in K$ by
\[
U_y:=\{f(y):f\in\mathcal P_i\},
\]
satisfy $U_{\alpha_i}=K$ and
\[
\dim_F U_y=\frac{\ell}{s}
\quad\text{for every }y\in\Omega\setminus\{\alpha_i\}.
\]
Moreover, for every $s$-subset
$B=\{b_1,\ldots,b_s\}\subseteq\Omega\setminus\{\alpha_i\}$
and every choice of $a_u\in U_{b_u}$,
$1\le u\le s$, there is an $f\in\mathcal P_i$
such that $f(b_u)=a_u$
for all $u$.
\end{lemma}

\begin{proof}
By \cref{thm:GW} and the dual form of RS codes, choose
polynomials $f_1,\ldots,f_\ell\in K[x]_{<s}$ whose values at
$\alpha_i$ form an $F$-basis of $K$ and whose helper
evaluation spans have total dimension at most $(n-1)\ell/s$.
Nonzero dual multipliers do not change these dimensions. Set
\[
\mathcal P_i=\Span_F\{f_1,\ldots,f_\ell\}.
\]
Then $\dim_F\mathcal P_i=\ell$ and $U_{\alpha_i}=K$.

For each helper point $y\in\Omega\setminus\{\alpha_i\}$, put
\[
t_y=\dim_F U_{y}.
\]
The selected polynomials yield a repair scheme with bandwidth
$\sum_y t_y$. The preceding upper bound and the cut-set
bound therefore force
\[
\sum_{y\in\Omega\setminus\{\alpha_i\}}t_y=\frac{(n-1)\ell}{s}.
\]

Let $B$ be any $s$-subset of the helper points. A polynomial of degree
less than $s$ is determined by its values on $B$, so evaluation on
$B$ is injective on $\mathcal P_i$. Its image has dimension
$\ell$ and is contained in the space of all tuples with values in
$U_y$ at each $y\in B$. Hence
\begin{equation}
\ell\le\sum_{y\in B}t_y.
\label{eq:s-subset-local}
\end{equation}
Summing the nonnegative differences over all such $B$ gives
\[\sum_{\substack{B\subseteq\Omega\setminus\{\alpha_i\}\\ |B|=s}}
\left(\sum_{y\in B}t_y-\ell\right)=\binom{n-2}{s-1}\frac{(n-1)\ell}{s}
-\binom{n-1}{s}\ell
=0.\]
Thus $\sum_{y\in B}t_y=\ell$ for every $B$.

Given distinct helper points $y,y'$, choose an
$(s-1)$-subset $B'$ of the remaining helper points; this is possible
since $n-3\ge s-1$. Comparing
$B'\cup\{y\}$ and $B'\cup\{y'\}$ shows that $t_y=t_{y'}$.
The total sum then yields $t_y=\ell/s$ for every helper point.

Finally, fix $B=\{b_1,\ldots,b_s\}$. The space of all tuples with
entries in the corresponding $U_{b_u}$ has dimension $\ell$.
By injectivity, the evaluation tuples from $\mathcal P_i$ form a
subspace of the same dimension. The two spaces therefore coincide,
so every prescribed tuple is attained.
\end{proof}

We next derive relations among the helper evaluation spaces.
For a field $L$, the projective line $\Pone(L)$ consists of nonzero
vectors in $L^2$ modulo nonzero scalar multiples. For column vectors
$u=(u_0,u_1)^{\mathsf T}\in L^2$ and $v=(v_0,v_1)^{\mathsf T}\in L^2$, write
\[
\det(u,v):=\det[\,u\ v\,]=u_0v_1-u_1v_0.
\]
Here $\mathrm{GL}_2(L)$ is the group of invertible $2\times2$ matrices
over $L$, and $\mathrm{PGL}_2(L)$ is its quotient by nonzero scalar matrices.
For four
distinct points $x,z,a,b\in\Pone(\overline F)$, choose nonzero
column-vector representatives, denoted by the same symbols, and define
\[
\chi(x,z;a,b)
=
\frac{\det(x,b)\det(z,a)}
     {\det(x,a)\det(z,b)}.
\]
This is independent of the chosen representatives.  If
$T\in\mathrm{PGL}_2(\overline F)$ is induced by
$A\in\mathrm{GL}_2(\overline F)$, then
$\det(Au,Av)=\det(A)\det(u,v)$, and hence
\[
\chi(x,z;a,b)
=
\chi\bigl(T(x),T(z);T(a),T(b)\bigr).
\]
Under the affine chart $t\leftrightarrow[t:1]$, with
$\infty=[1:0]$, this becomes
\[
\chi(x,z;a,b)
=
\frac{(x-b)(z-a)}{(x-a)(z-b)}
\]
when all four points are affine.

Fix an $s$-subset
$B=\{b_1,\ldots,b_s\}\subseteq\Omega\setminus\{\alpha_i\}$.
For $u\in[s]$, let $h_u(x)\in K[x]_{<s}$ be the Lagrange basis polynomial
\[
h_u(x)
=
\prod_{\substack{1\le t\le s\\ t\ne u}}
\frac{x-b_t}{b_u-b_t}.
\]
Fix $y\in K$. For $f\in\mathcal P_i$,
$f(y)=\sum_{u=1}^sh_u(y)f(b_u)$. By
\cref{lem:evaluation-rigidity}, every tuple of values from the
spaces $U_{b_u}$ is realized by a polynomial in
$\mathcal P_i$. Therefore
\begin{equation}
U_y=\sum_{u=1}^sh_u(y)U_{b_u}.
\label{eq:space-interpolation}
\end{equation}
The sum on the right is a sum of $F$-subspaces of $K$.
Here $h_u(y)$ is a scalar and
\[
h_u(y)U_{b_u}=\{h_u(y)a:a\in U_{b_u}\}
\]
is its scalar multiple of the $F$-subspace $U_{b_u}$.

At a helper point outside $B$, equality of dimensions forces
the summands in \eqref{eq:space-interpolation} to coincide.
This is the helper-space alignment used in
\cite[proof of Theorem~5]{TYB2017}. Comparing the identities at
two helper points yields the following relation.

\begin{lemma}
\label{lem:helper-alignment}
Let $B=\{b_1,\ldots,b_s\}\subseteq\Omega\setminus\{\alpha_i\}$ and let
$z\in\Omega\setminus(B\cup\{\alpha_i\})$. Then
\[
h_1(z)U_{b_1}=\cdots=h_s(z)U_{b_s}.
\]
Denote this common subspace by $A=U_z$. For every
$x\in\Omega\setminus(B\cup\{\alpha_i,z\})$ and every distinct
$u,v\in[s]$,
\begin{equation}
\chi(x,z;b_u,b_v)A=A.
\label{eq:alignment-cr}
\end{equation}
\end{lemma}

\begin{proof}
Taking $y=z$ in \eqref{eq:space-interpolation}, each nonzero scalar
multiple $h_u(z)U_{b_u}$ has dimension $\ell/s$, while their sum
$U_{z}$ also has dimension $\ell/s$.  Hence
\[
h_1(z)U_{b_1}=\cdots=h_s(z)U_{b_s}=U_{z}=A.
\]
The same argument at $x$ yields
\[
h_u(x)U_{b_u}=h_v(x)U_{b_v}.
\]
Using $h_u(z)U_{b_u}=h_v(z)U_{b_v}=A$, we obtain
\[
\frac{h_u(x)h_v(z)}{h_v(x)h_u(z)}A=A.
\]
A direct calculation using the Lagrange basis shows that
\[
\frac{h_u(x)h_v(z)}{h_v(x)h_u(z)}
=
\frac{(x-b_v)(z-b_u)}{(x-b_u)(z-b_v)}
=
\chi(x,z;b_u,b_v),
\]
which proves \eqref{eq:alignment-cr}.
\end{proof}

For a finite set $X\subseteq\Pone(\overline F)$, define
\[
\crf_F(X)
=
F\bigl(
\chi(x,z;a,b):x,z,a,b\in X\text{ distinct}
\bigr),
\]
with the convention $\crf_F(X)=F$ when $|X|\le3$.  We call $\crf_F(X)$ the
\emph{cross-ratio field} of $X$ over $F$.

Fix distinct $a,b,c\in\Pone(\overline F)$. For $y\in\Pone(\overline F)$, define
\[
\phi_{a,b,c}(y)
=
\bigl[
\det(y,b)\det(c,a):
\det(y,a)\det(c,b)
\bigr].
\]
The two homogeneous coordinate functions are linearly independent linear forms in $y$, so
$\phi_{a,b,c}$ is a projective transformation.  It satisfies
\[
\phi_{a,b,c}(a)=\infty,
\qquad
\phi_{a,b,c}(b)=0,
\qquad
\phi_{a,b,c}(c)=1.
\]
For $y\notin\{a,b,c\}$, its affine coordinate is
\begin{equation}
\phi_{a,b,c}(y)
=
\frac{\det(y,b)\det(c,a)}
     {\det(y,a)\det(c,b)}
=
\chi(y,c;a,b).
\label{eq:phi-cross-ratio}
\end{equation}
For $y\ne a$, we henceforth also use $\phi_{a,b,c}(y)$ to denote
its affine coordinate.

Although the cross-ratio field is defined using all quadruples, it is
generated after fixing any three points.

\begin{lemma}
\label{lem:three-anchor}
Let $X\subseteq\Pone(\overline F)$ contain distinct points $a,b,c$ and
write $\phi=\phi_{a,b,c}$. Then
\[
\crf_F(X)
=
F\bigl(\phi(y):y\in X\setminus\{a,b,c\}\bigr)
=
F\bigl(\chi(y,c;a,b):y\in X\setminus\{a,b,c\}\bigr).
\]
\end{lemma}

\begin{proof}
Let
\[
L=F\bigl(\phi(y):y\in X\setminus\{a,b,c\}\bigr).
\]
By \eqref{eq:phi-cross-ratio}, every generator of $L$ is a cross-ratio
of points of $X$, so $L\subseteq\crf_F(X)$.  Conversely,
$\phi(X)\subseteq\Pone(L)$ because
$\phi(a)=\infty$, $\phi(b)=0$, $\phi(c)=1$, and all remaining images
belong to $L$.  For any four distinct points $x_1,x_2,x_3,x_4\in X$,
projective invariance implies
\[
\chi(x_1,x_2;x_3,x_4)
=
\chi\bigl(\phi(x_1),\phi(x_2);\phi(x_3),\phi(x_4)\bigr)\in L.
\]
Thus every generator of $\crf_F(X)$ lies in $L$, proving the result.
\end{proof}

We now apply this field to the helper configuration
$\Omega\setminus\{\alpha_i\}$.  The use of the stabilizer of an aligned
repair space to obtain a field-divisibility condition already appears
in the lower-bound argument of Tamo, Ye, and Barg
\cite[proof of Theorem~5]{TYB2017}.  Here the relevant field is identified
with the cross-ratio field of the helper points.

\begin{proposition}
\label{prop:helper-field-divides}
For every choice of $B$ and $z$ as in \cref{lem:helper-alignment}, the
corresponding space $A$ is a
$\crf_F(\Omega\setminus\{\alpha_i\})$-linear subspace of $K$.
Consequently,
\[
\bigl[\crf_F(\Omega\setminus\{\alpha_i\}):F\bigr]
\mid\frac{\ell}{s}.
\]
\end{proposition}

\begin{proof}
Let $L=F(\chi(x,z;b_u,b_v):
x\in\Omega\setminus(B\cup\{\alpha_i,z\}),
1 \leq u\ne v \leq s).
$
By \cref{lem:helper-alignment}, multiplication by each of these scalars
maps $A$ onto itself.  Consider
\[
S_A=\{\lambda\in K:\lambda A\subseteq A\}.
\]
Since $\dim_F A=\ell/s>0$, the space $A$ is nonzero.
The set $S_A$ contains $F$ and is closed under addition and
multiplication.  If $0\ne\lambda\in S_A$, then
$\dim_F(\lambda A)=\dim_F A$, so $\lambda A=A$ and
$\lambda^{-1}A=A$.  Hence $S_A$ is a subfield of $K$, and
$L\subseteq S_A$.

We claim that
\[
L=\crf_F(\Omega\setminus\{\alpha_i\}).
\]
The inclusion $L\subseteq\crf_F(\Omega\setminus\{\alpha_i\})$ is
immediate.  For the reverse inclusion, let
$\phi=\phi_{b_1,b_2,z}$.  By \cref{lem:three-anchor}, it is enough to
show that $\phi(y)\in L$ for every
$y\in\Omega\setminus\{\alpha_i,b_1,b_2,z\}$.  If $y\notin B$, then
\[
\phi(y)=\chi(y,z;b_1,b_2)\in L.
\]
Now let $y=b_u$ with $u\ge3$.  Since
\[
\bigl|\Omega\setminus(B\cup\{\alpha_i,z\})\bigr|=k-2\ge1,
\]
choose $w$ in this set.  Then
\[
\phi(w)=\chi(w,z;b_1,b_2)\in L,
\qquad
c_u:=\chi(w,z;b_1,b_u)\in L.
\]
Using
$\phi(b_1)=\infty$, $\phi(b_2)=0$, and $\phi(z)=1$, projective
invariance gives
\[
c_u
=
\frac{\phi(w)-\phi(b_u)}{1-\phi(b_u)}.
\]
If $c_u=1$, then $\phi(w)=1=\phi(z)$, contrary to the
injectivity of $\phi$ and $w\ne z$. Hence $c_u\ne1$, and
\[
\phi(b_u)=\frac{\phi(w)-c_u}{1-c_u}\in L.
\]
This proves the claim.  Therefore $A$ is
$\crf_F(\Omega\setminus\{\alpha_i\})$-linear.  Finally,
$\dim_F A=\ell/s$, so
\[
\frac{\ell}{s}
=
\bigl[\crf_F(\Omega\setminus\{\alpha_i\}):F\bigr]\dim_{\crf_F(\Omega\setminus\{\alpha_i\})}A,
\]
which proves the divisibility.
\end{proof}

The failed coordinate imposes one further constraint. Fix $B,z,A$ as in
\cref{lem:helper-alignment}. Put $\rho_1=1$
and, for $2\le u\le s$, define
\begin{equation}
\rho_u
=
\frac{h_u(\alpha_i)h_1(z)}{h_1(\alpha_i)h_u(z)}
=
\chi(\alpha_i,z;b_u,b_1).
\label{eq:rho-def}
\end{equation}

\begin{proposition}
\label{prop:failure-field-growth}
The elements $\rho_1,\ldots,\rho_s$ are linearly independent over
$\crf_F(\Omega\setminus\{\alpha_i\})$. Consequently,
\[
\bigl[
\crf_F(\Omega):
\crf_F(\Omega\setminus\{\alpha_i\})
\bigr]\ge s.
\]
\end{proposition}

\begin{proof}
Setting $y=\alpha_i$ in \eqref{eq:space-interpolation} yields a
surjective map
\[
\bigoplus_{u=1}^sU_{b_u}\longrightarrow K,
\qquad
(v_1,\ldots,v_s)\longmapsto
\sum_{u=1}^sh_u(\alpha_i)v_u.
\]
Both spaces have $F$-dimension $\ell$, so this map is an isomorphism.
Equivalently,
\begin{equation}
K=\bigoplus_{u=1}^sh_u(\alpha_i)U_{b_u}.
\label{eq:failure-sum-cr}
\end{equation}
Set
\[
A'=\frac{h_1(\alpha_i)}{h_1(z)}A.
\]
By \cref{prop:helper-field-divides}, $A'$ is
$\crf_F(\Omega\setminus\{\alpha_i\})$-linear.  Since
$A=h_u(z)U_{b_u}$, \eqref{eq:rho-def} gives
\[
\rho_uA'=h_u(\alpha_i)U_{b_u},
\qquad 1\le u\le s.
\]
Thus \eqref{eq:failure-sum-cr} becomes
\[
K=\bigoplus_{u=1}^s\rho_uA'.
\]
If
\[
\sum_{u=1}^sc_u\rho_u=0,
\qquad
c_u\in\crf_F(\Omega\setminus\{\alpha_i\}),
\]
choose $0\ne a\in A'$. Multiplying the relation by $a$ and using
$c_ua\in A'$, the direct sum forces $\rho_uc_ua=0$ for every $u$.
Since $\rho_ua\ne0$, we obtain $c_u=0$.  Therefore the $\rho_u$ are
linearly independent over the helper cross-ratio field.  Since
$\rho_u\in\crf_F(\Omega)$, the relative degree is at least $s$.
\end{proof}

We have proved that, whenever a coordinate has cut-set-optimal
full-helper repair,
\[
\bigl[\crf_F(\Omega\setminus\{\alpha_i\}):F\bigr]
\mid\frac{\ell}{s},
\qquad
\bigl[
\crf_F(\Omega):\crf_F(\Omega\setminus\{\alpha_i\})
\bigr]\ge s.
\]
These two degree constraints will be used repeatedly below.

\subsection{A General Lower Bound}
To combine the degree constraints for different deleted
coordinates, we first record a simple identity for cross-ratio fields.

\begin{lemma}
\label{lem:two-deletion}
Let $A,B\subseteq\Pone(\overline F)$ be finite sets with
$|A\cap B|\ge3$. Then
\[
\crf_F(A\cup B)=\crf_F(A)\crf_F(B),
\]
where the right-hand side is the compositum inside $\overline F$.
In particular, if
\[
X=\{\alpha_1,\ldots,\alpha_m\}\subseteq\Pone(\overline F),
\qquad m\ge5,
\]
then, for distinct $i,j\in[m]$,
\[
\crf_F(X)
=
\crf_F(X\setminus\{\alpha_i\})
\crf_F(X\setminus\{\alpha_j\}).
\]
\end{lemma}

\begin{proof}
The inclusion
$\crf_F(A)\crf_F(B)\subseteq\crf_F(A\cup B)$ follows from the
definition.  Choose distinct $a,b,c\in A\cap B$ and let
$\phi=\phi_{a,b,c}$.  By \cref{lem:three-anchor},
\begin{align*}
\crf_F(A)
&=F\bigl(\phi(y):y\in A\setminus\{a,b,c\}\bigr),\\
\crf_F(B)
&=F\bigl(\phi(y):y\in B\setminus\{a,b,c\}\bigr),\\
\crf_F(A\cup B)
&=F\bigl(\phi(y):y\in(A\cup B)\setminus\{a,b,c\}\bigr).
\end{align*}
The last field is generated by the first two, proving equality.  The
special case follows by taking
$A=X\setminus\{\alpha_i\}$ and
$B=X\setminus\{\alpha_j\}$.
\end{proof}

For a configuration admitting cut-set-optimal full-helper repair of every coordinate, the degree lost when one evaluation point is deleted is at least $s$. This leads to the following
arithmetic quantity.

\begin{definition}
\label{def:Gs}
For integers $s\ge2$, define $G_s(0)=1$.
For an integer $m\ge1$, minimize over integer tuples $(g_1,\ldots,g_m)$:
\[
G_s(m)
=
\min\left\{
\prod_{j=1}^m g_j:
 g_j\ge s,\ \gcd(g_i,g_j)=1\text{ for }i\ne j
\right\}.
\]
\end{definition}

Let $p_j$ be the $j$-th prime and write
\[
\primorial{m}=\prod_{j=1}^{m}p_j,
\qquad
\primorial{0}=1.
\]
In particular, $G_2(m)=\primorial{m}$, and
$G_s(m)\ge\primorial{m}$ for every $s\ge2$.

\begin{theorem}
\label{thm:local-G-lower}
Let $F=\F_q$ and $K=\F_{q^\ell}$.  Let $k\ge3$,
$s=n-k \ge2$,  and let $\Omega\subseteq K$ be an evaluation
set of size $n$.  Suppose that $\RS_K(n,k,\Omega)$ admits
cut-set-optimal $F$-linear repair of every coordinate from all other
$n-1$ coordinates. Then
\[
\ell\ge sG_s(n).
\]
More precisely, there exist pairwise coprime integers
$g_1,\ldots,g_n\ge s$ such that
\[
\prod_{i=1}^n g_i
\mid[\crf_F(\Omega):F]
\mid\frac{\ell}{s}.
\]
\end{theorem}

\begin{proof}
Put
\[
m=[\crf_F(\Omega):F],
\qquad
m_i=[\crf_F(\Omega\setminus\{\alpha_i\}):F],
\qquad
g_i=\frac{m}{m_i}.
\]
By \cref{prop:helper-field-divides,prop:failure-field-growth},
\[
g_i\ge s,
\qquad
m_i\mid\frac{\ell}{s}.
\]
For distinct $i,j$, \cref{lem:two-deletion} gives
\[
\crf_F(\Omega)
=
\crf_F(\Omega\setminus\{\alpha_i\})
\crf_F(\Omega\setminus\{\alpha_j\}).
\]
Finite subfields of $\overline F$ are determined by their degrees over
$F$, and the degree of their compositum is the least common multiple.
Hence
\[
m=\lcmop(m_i,m_j).
\]
Since $m_i,m_j\mid\ell/s$, this also shows $m\mid\ell/s$.  Moreover,
\[
g_i=\frac{m_j}{\gcd(m_i,m_j)},
\qquad
g_j=\frac{m_i}{\gcd(m_i,m_j)},
\]
so $\gcd(g_i,g_j)=1$.  Thus the $g_i$ are pairwise coprime.  Each
$g_i$ divides $m$, so their product divides $m$:
\[
\prod_{i=1}^ng_i\mid m\mid\frac{\ell}{s}.
\]
By the definition of $G_s(n)$,
$\prod_i g_i\ge G_s(n)$, and therefore $\ell\ge sG_s(n)$.
\end{proof}

A bound for arbitrary $d$ follows from the full-helper result by
puncturing to one failed coordinate and its helpers.

\begin{corollary}
\label{cor:one-puncturing}
Let $k\ge3$ and $s=d+1-k\ge2$. If an $(n,k)$ RS code admits
cut-set-optimal $F$-linear repair of every coordinate from every
$d$-helper set, then
\[
\ell\ge sG_s(d+1).
\]
\end{corollary}

\begin{proof}
Puncture the code to any $d+1$ coordinates.  For every coordinate in
the punctured code, the other $d$ coordinates form an allowed helper
set in the original code.  Thus the punctured $(d+1,k)$ RS code admits
cut-set-optimal full-helper repair of every coordinate.  Applying
\cref{thm:local-G-lower} with $n=d+1$ proves the claim.
\end{proof}

\begin{remark}
Tamo, Ye, and Barg \cite[Theorem~5]{TYB2017} proved that every scalar
linear MDS code admitting cut-set-optimal $F$-linear repair of every
coordinate from every $d$-helper set satisfies
\[
\ell\ge\prod_{r=1}^{k-1}p_r.
\]
Any $d+1$ pairwise coprime integers at least $s\ge2$ collectively
contain at least $d+1$ distinct prime divisors, and hence
\[
G_s(d+1)\ge\prod_{r=1}^{d+1}p_r > \prod_{r=1}^{k-1}p_r.
\]
Thus, for RS codes, the general lower bound in \cref{cor:one-puncturing}
is stronger than that of \cite[Theorem~5]{TYB2017}.
\end{remark}

\subsection{An Improved Lower Bound for \texorpdfstring{$2d\geq n+3$}{2d >= n+3}}
\label{subsec:two-d-lower}

The preceding bound uses only one set of $d+1$ coordinates.  We now use
the constraints from all failed-coordinate/helper-set pairs.  When
$2d\geq n+3$, the resulting cross-ratio degrees admit an LCM
representation, leading to the lower bound $\ell\ge s\Delta_s(n,d)$.
We first record the degree constraints obtained from every punctured
repair configuration.

\begin{proposition}
\label{prop:global-crossratio}
Let $k\ge3$ and $s=d+1-k\ge2$. Suppose that
$\RS_K(n,k,\Omega)$ admits cut-set-optimal $F$-linear repair of every
coordinate from every $d$-helper set. Then, for every $x\in\Omega$ and
every $H\subseteq\Omega\setminus\{x\}$ with $|H|=d$,
\[
[\crf_F(H):F]\mid\frac{\ell}{s},
\qquad
[\crf_F(H\cup\{x\}):\crf_F(H)]\ge s.
\]
Moreover,
\[
[\crf_F(\Omega):F]\mid\frac{\ell}{s}.
\]
\end{proposition}

\begin{proof}
Puncture to $H\cup\{x\}$.  The resulting length-$(d+1)$ RS code admits
cut-set-optimal full-helper repair of each of its coordinates.  Applying
\cref{prop:helper-field-divides,prop:failure-field-growth} with target
$x$ yields
\[
[\crf_F(H):F]\mid\frac{\ell}{s},
\qquad
[\crf_F(H\cup\{x\}):\crf_F(H)]\ge s.
\]

It remains to prove the global divisibility. Since $d\ge4$,
every four-point subset of $\Omega$ lies in some $d$-subset $H$.
Thus $\crf_F(\Omega)$ is the compositum of the fields
$\crf_F(H)$ over all $d$-subsets. By the finite-field
compositum formula and the first part,
\[
[\crf_F(\Omega):F]
=\lcmop_{\substack{H\subseteq\Omega\\ |H|=d}}
[\crf_F(H):F]
\mid\frac{\ell}{s}.
\]
\end{proof}

We next define the arithmetic quantity that will also appear in the
construction.

\begin{definition}
\label{def:Delta}
Let $s,N,t$ be integers with $s\ge2$ and $1\le t<N$. A tuple of positive integers
$(a_1,\ldots,a_N)$ is \emph{feasible for $(N,t,s)$} if
\[
\frac{\lcmop(a_j:j\in J\cup\{i\})}
     {\lcmop(a_j:j\in J)}\ge s
\]
for every $i\in[N]$ and every $t$-subset
$J\subseteq[N]\setminus\{i\}$.  Define $\Delta_s(N,t)$ as the minimum
of $\lcmop(a_1,\ldots,a_N)$ over all feasible tuples.
\end{definition}

Thus $\Delta_s(N,t)$ is the minimum lcm compatible with all the above
relative-degree inequalities.  Any $N$ pairwise coprime integers at
least $s$ form a feasible tuple, so
\begin{equation}
\Delta_s(N,t)\le G_s(N).
\label{eq:Delta-G-basic}
\end{equation}

The following elementary lemma converts identities of the form
$m_{S\cap T}=\lcmop(m_S,m_T)$ into a representation by coordinate
labels.

\begin{lemma}
\label{lem:canonical-lcm-representation}
Let $0\le r<n$, and let $m_T$ be a positive integer for every
$T\subseteq[n]$ with $|T|\le r$.  The following are equivalent:
\begin{enumerate}
\item for every such $S,T$,
\[
m_{S\cap T}=\lcmop(m_S,m_T);
\]
\item there exist positive integers $b_1,\ldots,b_n$ such that
\[
m_T=\lcmop(b_i:i\notin T)
\]
for every $|T|\le r$.
\end{enumerate}
One possible choice is
\[
b_i=\gcd(m_R:|R|\le r,\ i\notin R).
\]
\end{lemma}

\begin{proof}
Assume the first condition.  If $R\subseteq V$, then
\[
m_R=\lcmop(m_R,m_V),
\]
so $m_V\mid m_R$.  Fix a prime $p$.
For a positive integer $a$, $v_p(a)$ denotes the
exponent of $p$ in the prime factorization of $a$. Write
\[
e_T=v_p(m_T),
\qquad
\lambda_i=\min_{\substack{|R|\le r\\ i\notin R}}e_R.
\]
For $i\notin T$, we have $\lambda_i\le e_T$.  We claim that equality
holds for at least one such $i$.  Otherwise, for each $i\notin T$ choose
$R_i$ with $i\notin R_i$ and $e_{R_i}=\lambda_i<e_T$, and put
\[
R=\mathop{\bigcap}\limits_{i\notin T}R_i.
\]
For every $i\notin T$, the point $i$ is missing from $R_i$, hence from
$R$; therefore $R\subseteq T$.  Repeated application of the first condition yields
\[
e_R=\max_{i\notin T}e_{R_i}<e_T.
\]
But $R\subseteq T$ implies $m_T\mid m_R$, hence $e_T\le e_R$, a
contradiction.  Thus
\[
e_T=\max_{i\notin T}\lambda_i.
\]
Since $v_p(b_i)=\lambda_i$, this is exactly the $p$-adic valuation
of $\lcmop(b_i:i\notin T)$.  As this holds for every prime,
\[
m_T=\lcmop(b_i:i\notin T).
\]

Conversely, if such a representation holds, then
\[
[n]\setminus(S\cap T)
=([n]\setminus S)\cup([n]\setminus T),
\]
and the desired lcm identity follows immediately.
\end{proof}

Suppose now that $2d\geq n+3$.  If $S,T\subseteq[n]$ have
$|S|,|T|\le n-d$, then at least three evaluation points remain after
deleting $S\cup T$.  This allows the preceding lcm representation to
be applied simultaneously to all relevant cross-ratio fields.

\begin{theorem}
\label{thm:two-d-lower}
Let $3\le k<d\le n-1$, put $s=d+1-k$, and suppose that
$\RS_K(n,k,\Omega)$ admits cut-set-optimal $F$-linear repair of every
coordinate from every $d$-helper set.  If
\[
2d\geq n+3,
\]
then there exist positive integers $b_1,\ldots,b_n$ such that, for any
$T\subseteq[n]$ with $|T|\le n-d$,
\begin{equation}
[\crf_F(\{\alpha_j:j\notin T\}):F]
=\lcmop(b_j:j\notin T)
\label{eq:two-d-lcm-system}
\end{equation}
and
\begin{equation}
\lcmop(b_1,\ldots,b_n)
=[\crf_F(\Omega):F]
\mid\frac{\ell}{s}.
\label{eq:two-d-global-lcm}
\end{equation}
Moreover, $(b_1,\ldots,b_n)$ is feasible for $(n,d,s)$.
Consequently,
\[
\ell\ge s\Delta_s(n,d),
\]
and hence
\[
\ellmin(n,k,d;q)\ge s\Delta_s(n,d).
\]
\end{theorem}

\begin{proof}
Put $r=n-d$.
For $T\subseteq[n]$ with $|T|\le r$, put
\[
X_T=\{\alpha_j:j\notin T\},
\qquad
m_T=[\crf_F(X_T):F].
\]
For two such sets $S,T$,
\[
|X_S\cap X_T|
=n-|S\cup T|
\ge n-2r
=2d-n
\ge3.
\]
Moreover,
$X_S\cup X_T=X_{S\cap T}$.  Hence \cref{lem:two-deletion} gives
\[
\crf_F(X_{S\cap T})=\crf_F(X_S)\crf_F(X_T),
\]
and therefore, since the fields are finite,
\[
m_{S\cap T}=\lcmop(m_S,m_T).
\]
By \cref{lem:canonical-lcm-representation}, there are positive integers
$b_1,\ldots,b_n$ satisfying
\eqref{eq:two-d-lcm-system}.  Taking $T=\varnothing$ yields the first
equality in \eqref{eq:two-d-global-lcm}, while \cref{prop:global-crossratio} shows that
$[\crf_F(\Omega):F]$ divides $\ell/s$.

It remains to verify feasibility.  Let $H\subseteq[n]$ have size $d$
and let $i\notin H$.  Put $T=[n]\setminus H$, so $|T|=r$ and $i\in T$.
Applying \eqref{eq:two-d-lcm-system} to $T$ and
$T\setminus\{i\}$ gives
\[
\frac{\lcmop(b_i,\{b_h:h\in H\})}
     {\lcmop(b_h:h\in H)}
=
[\crf_F(\{\alpha_h:h\in H\}\cup\{\alpha_i\}):
  \crf_F(\{\alpha_h:h\in H\})].
\]
The right-hand side is at least $s$ by
\cref{prop:global-crossratio}.  Thus the tuple is feasible, and
\[
\Delta_s(n,d)
\le\lcmop(b_1,\ldots,b_n)
=[\crf_F(\Omega):F]
\le\frac{\ell}{s},
\]
where the divisibility underlying the last inequality was
established in \eqref{eq:two-d-global-lcm}.
\end{proof}

\subsection{A Stronger Lower Bound for \texorpdfstring{$d(d+2)\geq 4(n+1)$}{d(d+2) >= 4(n+1)}}
\label{subsec:quadratic-lower}

We now establish the bound $\ell\ge sG_s(n)$ under
$d(d+2)\geq 4(n+1)$, extending beyond the range $2d\geq n+3$.
When $2d<n+3$, two $d$-subsets may intersect in fewer than three
points, so the preceding LCM representation is no longer guaranteed
by the overlap argument.  We therefore work directly with the degrees
of the cross-ratio fields and track how they change when individual
evaluation points are deleted.

For every $Y\subseteq\Omega$ with $|Y|\ge3$, write
\[
m(Y)=[\crf_F(Y):F],
\qquad
M=m(\Omega).
\]
For $i\in[n]$, put
\[
m_i=m(\Omega\setminus\{\alpha_i\}),
\qquad
g_i=\frac{M}{m_i}.
\]
Thus $g_i$ measures the part of the total degree that disappears when
$\alpha_i$ is deleted.  We call coordinate $i$ \emph{nondeficient} if
$g_i\ge s$, and \emph{deficient} if $g_i<s$.  Let
\[
\mathcal B=\{i:g_i<s\},
\qquad
b=|\mathcal B|.
\]
The nondeficient coordinates already contribute individual factors of
size at least $s$.  For a deficient coordinate, the additional factor
required by the repair constraint must come from prime powers shared
with other coordinates.  The next lemma makes this decomposition
precise.

\begin{lemma}
\label{lem:private-factorization}
The integers $g_1,\ldots,g_n$ are pairwise coprime.  Put
\[
M_0:=\frac{M}{\prod_{i=1}^n g_i}.
\]
Then $M_0$ is a positive integer satisfying
\[
M_0=\gcd(m_1,\ldots,m_n),\qquad
\bigcap_{i=1}^n\crf_F(\Omega\setminus\{\alpha_i\})
=\mathbb F_{q^{M_0}}.
\]
Moreover,
\begin{equation}
M\ge M_0G_s(n-b).
\label{eq:private-factor-lower}
\end{equation}
\end{lemma}

\begin{proof}
For distinct $i,j$, the sets $\Omega\setminus\{\alpha_i\}$
and $\Omega\setminus\{\alpha_j\}$ intersect in at least
three points.  By \cref{lem:two-deletion},
\[
M=\lcmop(m_i,m_j).
\]
Fix a prime $p\mid M$, and write
$E=v_p(M)$ and $e_i=v_p(m_i)$.  Thus
$E=\max\{e_i,e_j\}$ for every $i\ne j$, so at most one $e_i$ is
smaller than $E$.  As $v_p(g_i)=E-e_i$, at most one $g_i$ is divisible
by $p$.  Hence the $g_i$ are pairwise coprime and
$\prod_i g_i\mid M$.  Furthermore,
\[
\begin{aligned}
v_p(M_0)
&=E-\sum_{i=1}^n(E-e_i)
=\min_{1\le i\le n}e_i
=v_p(\gcd(m_1,\ldots,m_n)).
\end{aligned}
\]
The same equality is immediate for $p\nmid M$, proving the gcd
identity.  The finite-field intersection formula now yields
\[\bigcap_{i=1}^n\crf_F(\Omega\setminus\{\alpha_i\})
=\bigcap_{i=1}^n\mathbb F_{q^{m_i}}=\mathbb F_{q^{\gcd(m_1,\ldots,m_n)}}
=\mathbb F_{q^{M_0}}.
\]
Finally, the $n-b$ nondeficient $g_i$ are
pairwise coprime and at least $s$, whence
\[
M=M_0\prod_{i=1}^n g_i
\ge M_0G_s(n-b).
\]
\end{proof}

Although the $g_i$ are pairwise coprime, they need not be
coprime to $M_0$. The factor $M_0$ is common to the cross-ratio
fields obtained by deleting any one point.  The next lemma tracks
how much of this common degree is generated by an arbitrary subset.

\begin{lemma}
\label{lem:shared-degree-factorization}
For $Y\subseteq\Omega$ with $|Y|\ge3$, define
\[
D(Y)=\gcd(m(Y),M_0).
\]
Then
\[
D(Y)
=[\crf_F(Y)\cap\mathbb F_{q^{M_0}}:F].
\]
For every $Y\subseteq\Omega$ with $|Y|\ge4$,
\begin{equation}
m(Y)=D(Y)\prod_{\alpha_i\in Y}g_i.
\label{eq:residual-factorization-main}
\end{equation}
Moreover, $D(A)\mid D(B)$ whenever
$A\subseteq B$ and $|A|\ge3$; if $|A\cap B|\ge3$, then
\[
D(A\cup B)
=\lcmop(D(A),D(B)).
\]
In particular, for every $i$ and every $d$-subset
$H\subseteq\Omega\setminus\{\alpha_i\}$,
\begin{equation}
\frac{m(H\cup\{\alpha_i\})}{m(H)}
=g_i\frac{D(H\cup\{\alpha_i\})}{D(H)}.
\label{eq:residual-ratio-main}
\end{equation}
\end{lemma}

\begin{proof}
The finite-field intersection formula gives
\[
\crf_F(Y)\cap\mathbb F_{q^{M_0}}
=\mathbb F_{q^{\gcd(m(Y),M_0)}}
=\mathbb F_{q^{D(Y)}},
\]
which proves the claimed field interpretation.

To prove \eqref{eq:residual-factorization-main}, fix a prime
$p\mid M$ and write
\[
E=v_p(M),\qquad e_i=v_p(m_i),\qquad
C_p=v_p(M_0)=\min_i e_i.
\]
By the proof of \cref{lem:private-factorization}, at most one $e_i$
is smaller than $E$, and $v_p(g_i)=E-e_i$.  Also,
\[
v_p(D(Y))=\min\{v_p(m(Y)),C_p\}.
\]
If every $e_i=E$, then $C_p=E$ and $p$ divides none of the $g_i$;
both sides of \eqref{eq:residual-factorization-main} have
$p$-adic exponent $v_p(m(Y))$.
Otherwise, let $j$ be the unique index with $e_j=C_p<E$.
Then $v_p(g_j)=E-C_p$, while $v_p(g_i)=0$ for $i\ne j$.
If $\alpha_j\notin Y$, then
$Y\subseteq\Omega\setminus\{\alpha_j\}$, so
$v_p(m(Y))\le C_p$ and $v_p(D(Y))=v_p(m(Y))$.
If $\alpha_j\in Y$ and $|Y|\ge4$, then
\[
Y\cup(\Omega\setminus\{\alpha_j\})=\Omega,
\qquad
|Y\cap(\Omega\setminus\{\alpha_j\})|\ge3.
\]
Therefore \cref{lem:two-deletion} yields
\[
M=\lcmop(m(Y),m_j),
\]
so $v_p(m(Y))=E$ because $C_p<E$.  The exponent on the right of
\eqref{eq:residual-factorization-main} is then
$C_p+(E-C_p)=E$.  The equality is trivial for $p\nmid M$,
proving \eqref{eq:residual-factorization-main}.

If $A\subseteq B$, then $m(A)\mid m(B)$, and hence
$D(A)\mid D(B)$.  If
$|A\cap B|\ge3$, \cref{lem:two-deletion} gives
$m(A\cup B)=\lcmop(m(A),m(B))$.  For any prime $p$,
writing $c=v_p(M_0)$, we have
\begin{align*}
v_p(D(A\cup B))
&=\min\{\max(v_p(m(A)),v_p(m(B))),c\}\\
&=\max\{\min(v_p(m(A)),c),\min(v_p(m(B)),c)\}\\
&=v_p(\lcmop(D(A),D(B))).
\end{align*}
This proves the lcm identity.  Finally, $d\ge4$ ensures that both
$H$ and $H\cup\{\alpha_i\}$ satisfy
\eqref{eq:residual-factorization-main}.  Dividing the two
factorizations yields \eqref{eq:residual-ratio-main}.
\end{proof}

For a deficient coordinate $i$, the factor $g_i$ in
\eqref{eq:residual-ratio-main} is strictly smaller than $s$.  On the
other hand, \cref{prop:global-crossratio} requires
\[
\frac{m(H\cup\{\alpha_i\})}{m(H)}\ge s
\]
for every $d$-helper set $H$.  By \eqref{eq:residual-ratio-main},
this forces
$D(H\cup\{\alpha_i\})/D(H)\ge s/g_i>1$:
the missing factor must come from the shared part.  The next lemma
quantifies how many distinct prime divisors of $M_0$ are forced
in this way.  For a positive integer $N$, let $\omega(N)$ denote the
number of distinct prime divisors of $N$.

\begin{lemma}
\label{lem:shared-prime-separation}
Use the notation above.  Let $W\subseteq\Omega$ satisfy
$4\le w:=|W|\le d+1$.  For every deficient coordinate
$i\in\mathcal B$ with $\alpha_i\in W$, put
\[
W_i=W\setminus\{\alpha_i\},
\qquad
q_i(W)=\frac{D(W)}{D(W_i)}.
\]
Then $q_i(W)$ is a positive integer dividing $M_0$, and
\begin{equation}
\omega(q_i(W))\ge d-w+2.
\label{eq:shared-prime-count}
\end{equation}
If $i\ne j$ are deficient coordinates with
$\alpha_i,\alpha_j\in W$ and $w\ge5$, then
\[
\gcd(q_i(W),q_j(W))=1.
\]
Consequently, if $b\ge1$, then
\begin{equation}
\omega(M_0)\ge d-2,
\label{eq:shared-singleton-support}
\end{equation}
and, for every $1\le u\le\min\{b,d+1\}$,
\begin{equation}
\omega(M_0)
\ge u\bigl(d+2-\max\{5,u\}\bigr).
\label{eq:shared-multiple-support}
\end{equation}
\end{lemma}

\begin{proof}
Since $W_i\subseteq W$, \cref{lem:shared-degree-factorization} gives
$D(W_i)\mid D(W)$, so $q_i(W)$ is an integer.  Moreover
$D(W)\mid M_0$, and hence every prime divisor of $q_i(W)$
divides $M_0$.

We first prove \eqref{eq:shared-prime-count}.  Suppose, to the contrary,
that
\[
r:=\omega(q_i(W))\le d-w+1.
\]
Because $|W_i|=w-1\ge3$, all sets $W_i\cup\{y\}$ with
$y\in\Omega\setminus\{\alpha_i\}$ contain the same three-point
subset.  Repeated application of the lcm identity for $D$ yields
\[
\lcmop_{y\in\Omega\setminus\{\alpha_i\}}D(W_i\cup\{y\})
=
D(\Omega\setminus\{\alpha_i\}).
\]
By \cref{lem:private-factorization},
$M_0\mid m_i$.  Therefore
\[
D(\Omega\setminus\{\alpha_i\})
=
\gcd(m_i,M_0)
=
M_0,
\]
and hence
\begin{equation}
\lcmop_{y\in\Omega\setminus\{\alpha_i\}}D(W_i\cup\{y\})=M_0.
\label{eq:shared-cover-lcm}
\end{equation}

For each prime $p\mid q_i(W)$, choose a point
$y_p\ne\alpha_i$ such that
\[
v_p\bigl(D(W_i\cup\{y_p\})\bigr)
=
v_p(M_0);
\]
such a point exists by \eqref{eq:shared-cover-lcm}.  Let
\[
S=W_i\cup\{y_p:p\mid q_i(W)\}.
\]
Then
\[
|S|
\le
(w-1)+r
\le
(w-1)+(d-w+1)
=d.
\]
Since $S\subseteq\Omega\setminus\{\alpha_i\}$, we can extend $S$ to
a $d$-subset
\[
H\subseteq\Omega\setminus\{\alpha_i\}.
\]
We claim that $D(W)\mid D(H)$. If $p\mid q_i(W)$,
the choice of $y_p$ and $D(H)\mid M_0$ give
\[
v_p(D(H))=v_p(M_0)\ge v_p(D(W)).
\]
If $p\nmid q_i(W)$, then
\[
v_p(D(W))=v_p(D(W_i))\le v_p(D(H)),
\]
since $W_i\subseteq H$. Thus $D(W)\mid D(H)$.
Moreover, $H\cap W=W_i$ has at least three points, so
\[
D(H\cup\{\alpha_i\})
=\lcmop(D(H),D(W))=D(H).
\]

Using \eqref{eq:residual-ratio-main}, we now obtain
\[
\frac{m(H\cup\{\alpha_i\})}{m(H)}
=
g_i\frac{D(H\cup\{\alpha_i\})}{D(H)}
=
g_i
<s,
\]
which contradicts \cref{prop:global-crossratio}.  Hence
$\omega(q_i(W))\ge d-w+2$.

It remains to prove the separation assertion.  Let $i\ne j$ be
deficient coordinates in $W$, with $w\ge5$.  Since $|W_i\cap W_j|=w-2\ge3$, the lcm identity yields
\[
D(W)=\lcmop(D(W_i),D(W_j)).
\]
Since both $D(W_i)$ and $D(W_j)$ divide $D(W)$,
\[
\gcd(q_i(W),q_j(W))
=\frac{D(W)}{\lcmop(D(W_i),D(W_j))}=1.
\]

If $b\ge1$, choose a deficient coordinate $i$ and a four-point set
$W$ containing $\alpha_i$.  Then
\eqref{eq:shared-prime-count} gives
\[
\omega(M_0)
\ge
\omega(q_i(W))
\ge
d-2,
\]
which proves \eqref{eq:shared-singleton-support}.

Finally, fix $1\le u\le\min\{b,d+1\}$ and choose $u$ distinct
deficient coordinates. Extend their evaluation points to a set $W$ of size
\[
w=\max\{5,u\}.
\]
Such a set exists because $w\le d+1\le n$.
For these $u$ coordinates, the quotients $q_i(W)$ are pairwise
coprime, each divides $M_0$, and each has at least
$d-w+2$ distinct prime divisors.  Therefore
\[
\omega(M_0)
\ge u(d-w+2)
=u\bigl(d+2-\max\{5,u\}\bigr),
\]
which is \eqref{eq:shared-multiple-support}.
\end{proof}

Let
\[
c_s=\pi(s-1).
\]

We first record two elementary estimates for $G_s$.  The first compares
$G_s(n)$ with $G_s(n-b)$; the second supplies the prime products needed
in the proof of the degree bound.

\begin{lemma}
\label{lem:degree-completion}
For $s\ge2$, $n\ge1$, and $0\le b\le n$,
\begin{equation}
G_s(n)\le G_s(n-b)p_{n+c_s}^{\,b}.
\label{eq:degree-completion}
\end{equation}
Moreover,
\begin{equation}
G_s(n)\le\prod_{j=1}^{n}p_{c_s+j}\le \primorial{n+c_s}.
\label{eq:degree-prime-list}
\end{equation}
\end{lemma}

\begin{proof}
Choose pairwise coprime integers $a_1,\ldots,a_{n-b}\ge s$ whose
product is $G_s(n-b)$.  If some $a_j$ has a prime divisor $p\ge s$ and
$a_j\ne p$, then replacing $a_j$ by $p$ preserves pairwise coprimality
and the condition $a_j\ge s$, while decreasing the product.  Thus, in
a minimizing family, every member having a prime divisor at least $s$
is itself such a prime.

The primes
\[
p_{c_s+1},\ldots,p_{c_s+n}
\]
are the first $n$ primes at least $s$.  At most $n-b$ of them occur
among the $a_j$, so at least $b$ are unused.  Append any $b$ unused
primes to the family.  They remain coprime to all the $a_j$, and each
is at most $p_{n+c_s}$.  Therefore
\[
G_s(n)\le G_s(n-b)p_{n+c_s}^{\,b},
\]
which proves \eqref{eq:degree-completion}.  The first
$n$ primes at least $s$ form a feasible family, so
\[
G_s(n)\le\prod_{j=1}^np_{c_s+j}\le \primorial{n+c_s},
\]
proving \eqref{eq:degree-prime-list}.
\end{proof}

The preceding comparison bounds the loss from the $b$ deficient
coordinates, while \cref{lem:shared-prime-separation} forces additional
prime factors into $M_0$.  The following estimates show that
these shared factors compensate for the missing private factors.  We state only the comparisons needed below; their proof, including
the auxiliary prime-product estimates, is given in
Appendix~\ref{app:degree-arithmetic}.

\begin{lemma}
\label{lem:elementary-prime-product}
Let $s\ge2$, $n\ge1$, and $t\ge3$ be integers, and suppose
$n+c_s\le4t^2$. Then
\begin{equation}
G_s(n)<G_s(n-1)\primorial{t+1}.
\label{eq:completion-one}
\end{equation}
Moreover, for every $2\le b\le n$,
\begin{equation}
G_s(n)<G_s(n-b)\primorial{bt}.
\label{eq:completion-many}
\end{equation}
\end{lemma}

We now combine the structural and arithmetic estimates to obtain the
cross-ratio degree bound in the entire quadratic parameter range.

\begin{proposition}
\label{prop:quadratic-degree-system}
Suppose that $\RS_K(n,k,\Omega)$ admits cut-set-optimal $F$-linear
repair of every coordinate from every $d$-helper set, where
$3\le k<d\le n-1$ and $s=d+1-k$. If
\[
d(d+2)\geq 4(n+1),
\]
then
\[
[\crf_F(\Omega):F]\ge G_s(n).
\]
For $d\ge6$, the same conclusion holds under the more general condition
\[
n+c_s\le
\left(\left\lfloor\frac d2\right\rfloor+2\right)
\left\lceil\frac d2\right\rceil.
\]
\end{proposition}

\begin{proof}
Write $M=[\crf_F(\Omega):F]$, and let $b$ be the number of deficient
coordinates. If $n=d+1$, \cref{thm:local-G-lower} gives
$M\ge G_s(n)$. Hence assume $n\ge d+2$.
The case $d=4$ is impossible under the quadratic condition,
which would imply $n\le5$.

First suppose $d\ge6$. Put
\[
t=\left\lceil\frac d2\right\rceil,\qquad
r=\left\lfloor\frac d2\right\rfloor+2.
\]
Since $2\le s\le d-2$ and every prime other than $2$
is odd, we have $c_s=\pi(s-1)\le t-1$. In particular,
$d(d+2)\geq 4(n+1)$ implies
\[
\begin{aligned}
n+c_s
&\le\left\lfloor\frac{(d+1)^2}{4}\right\rfloor+t\\
&=rt.
\end{aligned}
\]
It therefore suffices to assume $n+c_s\le rt$, as in the
more general assertion. Moreover, $r\ge5$, $t\ge3$, $r+t=d+2$, $d-2\ge t+1$,
and $rt\le2t^2<4t^2$.

If $b=0$, \eqref{eq:private-factor-lower} immediately implies
$M\ge G_s(n)$. If $b=1$, \eqref{eq:shared-singleton-support}
yields
\[
\omega(M_0)\ge d-2\ge t+1.
\]
Consequently, by \eqref{eq:private-factor-lower} and
\eqref{eq:completion-one},
\[
M\ge M_0G_s(n-1)
\ge\primorial{t+1}G_s(n-1)>G_s(n).
\]
If $2\le b<r$, then $\max\{5,b\}\le r$, and
\eqref{eq:shared-multiple-support} with $u=b$ gives
\[
\omega(M_0)
\ge b(d+2-\max\{5,b\})\ge bt.
\]
Hence \eqref{eq:completion-many} implies
\[
M\ge M_0G_s(n-b)
\ge\primorial{bt}G_s(n-b)>G_s(n).
\]
Finally, if $b\ge r$, taking $u=r$ in
\eqref{eq:shared-multiple-support} yields
\[
\omega(M_0)\ge r(d+2-r)=rt.
\]
Therefore, by \eqref{eq:degree-prime-list},
\[
M\ge M_0\ge\primorial{rt}
\ge\primorial{n+c_s}\ge G_s(n).
\]
This proves the assertion for $d\ge6$.

It remains to consider $d=5$: the parameter condition and
$n\ge d+2$ force $n=7$, while $s\in\{2,3\}$ and $c_s\le1$.
For distinct $i,j\in[7]$, set
\[
Y_{ij}=\Omega\setminus\{\alpha_i,\alpha_j\},\qquad
h_{ij}=\frac{M_0}{D(Y_{ij})}.
\]
By \cref{lem:private-factorization,lem:shared-degree-factorization},
$D(\Omega\setminus\{\alpha_j\})=M_0$.
Thus, if $i\in\mathcal B$ and $j\ne i$, then
\eqref{eq:residual-ratio-main} and \cref{prop:global-crossratio}
imply
\[
s\le\frac{m(\Omega\setminus\{\alpha_j\})}{m(Y_{ij})}
=g_i h_{ij},
\]
so $h_{ij}>1$.

For two distinct unordered pairs $\{i,j\}$ and $\{u,v\}$,
the sets $Y_{ij}$ and $Y_{uv}$ intersect in at least three
points, while their union is either $\Omega$ or $\Omega$ with one
point removed. In either case its $D$-value is
$M_0$. The lcm identity in
\cref{lem:shared-degree-factorization} therefore gives
\[
M_0
=\lcmop(D(Y_{ij}),D(Y_{uv})),
\]
which implies $\gcd(h_{ij},h_{uv})=1$.
Thus the nontrivial $h_{ij}$ indexed by pairs meeting
$\mathcal B$ contribute distinct prime divisors of $M_0$.
Consequently,
\[
\omega(M_0)
\ge \binom72-\binom{7-b}{2}
=\begin{cases}
6,&b=1,\\
\ge11,&b\ge2.
\end{cases}
\]
The case $b=0$ again follows from \eqref{eq:private-factor-lower}.
If $b=1$, then $p_{7+c_s}\le p_8=19<\primorial{6}$, and
\eqref{eq:degree-completion} gives
\[
M\ge M_0G_s(6)
\ge\primorial{6}G_s(6)
>p_{7+c_s}G_s(6)\ge G_s(7).
\]
If $b\ge2$, \eqref{eq:degree-prime-list} gives directly
\[
M\ge M_0\ge\primorial{11}
>\primorial{7+c_s}\ge G_s(7).
\]
This completes the proof.
\end{proof}

Together with the fact that the cross-ratio degree divides $\ell/s$,
\cref{prop:quadratic-degree-system} yields the lower bound on
subpacketization in the same parameter range.

\begin{theorem}
\label{thm:quadratic-global-lower}
For every prime power $q$ and all $3\le k<d\le n-1$, if
\[
d(d+2)\geq 4(n+1),
\]
then
\[
\ellmin(n,k,d;q)\ge sG_s(n)\ge s\Delta_s(n,d).
\]
\end{theorem}

\begin{proof}
Let $F=\F_q$, $K=\F_{q^\ell}$, and let $\RS_K(n,k,\Omega)$ admit cut-set-optimal $F$-linear
repair of every coordinate from every $d$-helper set.
By \cref{prop:quadratic-degree-system},
\[
[\crf_F(\Omega):F]\ge G_s(n).
\]
By \cref{prop:global-crossratio}, this degree divides $\ell/s$.
It follows that $\ell\ge sG_s(n)$. Minimizing over all such codes
and using \eqref{eq:Delta-G-basic} completes the proof.
\end{proof}

The condition $d(d+2)\geq 4(n+1)$ contains the entire range
$2d\geq n+3$.  Indeed, $2d\geq n+3$ implies $n\le2d-3$, and since
$d\ge4$,
\[
d(d+2)-4(n+1)
\ge(d-2)(d-4)\ge0.
\]

\section{Construction and Minimum Subpacketization}
\label{sec:construction}
In this section, we construct RS codes with subpacketization
$s\Delta_s(n,d)$ for every admissible choice of parameters.
Combining this construction with the lower bound from
Section~\ref{sec:lower-bounds} determines the minimum
subpacketization when $d(d+2)\geq 4(n+1)$.
We then derive an exact threshold for a closed-form expression
for $G_s(m)$ and a finite formula valid for all $m$.

\subsection{Construction from Feasible Degree Tuples}

The feasibility condition defining $\Delta_s(n,d)$ is exactly the
relative-degree condition needed to apply \cref{thm:Qiu-local}.

Indeed, for a failed coordinate $i$ and a helper set $H$, if
$[F(\alpha_j):F]=a_j$ and $F_H=F(\alpha_h:h\in H)$, then the compositum
rule for finite fields gives
\[
[F_H(\alpha_i):F_H]
=\frac{\lcmop(a_i,\{a_h:h\in H\})}
       {\lcmop(a_h:h\in H)}.
\]
Thus feasibility guarantees that this relative degree is at least $s$.
Consequently, \cref{thm:Qiu-local} provides the required repair
subspace for every such pair $(i,H)$.  A dual check polynomial vanishing
at the nonhelpers then restricts the parity-check equations to
$H\cup\{i\}$.

\begin{theorem}
\label{thm:degree-construction}
Let $q$ be a prime power and $F=\F_q$.
Let $3\le k<d\le n-1$ and set $s=d+1-k\ge2$.
Suppose that positive integers $a_1,\ldots,a_n$ form a feasible
$(n,d,s)$-tuple.  Put
\[
 M=\lcmop(a_1,\ldots,a_n),
 \qquad
 E=\F_{q^M},
 \qquad
 K=\F_{q^{sM}}.
\]
Choose $\alpha_i\in E$ with $[F(\alpha_i):F]=a_i$ and set
\[
 C=\RS_K(n,k,\{\alpha_1,\ldots,\alpha_n\}).
\]
Then $C$ repairs every failed coordinate from every $d$-element helper
set by an $F$-linear scheme meeting the cut-set bound.  Its
subpacketization is $\ell=sM$.  Consequently,
\[
 \ellmin(n,k,d;q)\le s\Delta_s(n,d).
\]
\end{theorem}

\begin{proof}
Feasibility first implies that $a_1,\ldots,a_n$ are pairwise distinct.
Indeed, if $a_i=a_j$ for some $i\ne j$, choose a $d$-helper set
$H$ that contains $j$ but not $i$.  Then
\[
\frac{\lcmop(a_i,\{a_h:h\in H\})}
{\lcmop(a_h:h\in H)}=1,
\]
contradicting $s\ge2$.

Since every $a_i$ divides $M$, the field $E=\F_{q^M}$ contains the
subfield $\F_{q^{a_i}}$, and hence contains elements of exact degree
$a_i$ over $F$.  Choose such an element $\alpha_i$ for every $i$.
Because the degrees $a_i$ are distinct, the chosen evaluation points
are automatically pairwise distinct.

Fix a failed coordinate $i\in[n]$ and a helper set
$H\subseteq[n]\setminus\{i\}$ with $|H|=d$.
The following fields and subspaces refer to this fixed pair $(i,H)$. Put
\[
L_H=\lcmop(a_h:h\in H),
\qquad
L=\lcmop(L_H,a_i).
\]
Since the compositum of finite subfields has degree equal to the least
common multiple of their degrees,
\[
F(\alpha_h:h\in H)=\F_{q^{L_H}},
\qquad
F(\alpha_h:h\in H,\alpha_i)=\F_{q^{L}}.
\]
Since $L_H\mid L\mid M$, we have the tower
\[
F_H:=\F_{q^{L_H}}
\subset E_H:=\F_{q^L}
\subset K_0:=\F_{q^{sL}}
\subseteq K.
\]
Here $E_H=F_H(\alpha_i)$, and feasibility gives
$[E_H:F_H]=L/L_H\ge s$. We first work in $K_0$
because $[K_0:E_H]=s$, as required by
\cref{thm:Qiu-local}. That theorem provides an
$F_H$-subspace $S_0\subseteq K_0$ such that
\[
\dim_{F_H}S_0=\frac{L}{L_H},
\qquad
K_0=\bigoplus_{u=0}^{s-1}\alpha_i^uS_0.
\]

Let
\[
r=[K:K_0]=\frac{M}{L},
\]
choose a $K_0$-basis
$\gamma_1,\ldots,\gamma_{r}$ of $K$, and define
\[
S=\bigoplus_{v=1}^{r}\gamma_vS_0.
\]
Then
\[
K
=\bigoplus_{v=1}^{r}\gamma_vK_0
=\bigoplus_{u=0}^{s-1}\alpha_i^uS,
\]
and
\[
\dim_{F_H}S
=r\frac{L}{L_H}
=\frac{M}{L_H}
=\frac{[K:F_H]}{s}.
\]

We next realize this local decomposition by dual parity checks supported
on the failed coordinate and the chosen helpers.  Let
\[
Z=[n]\setminus(H\cup\{i\}),
\qquad
\Lambda_{i,H}(x)=\prod_{z\in Z}(x-\alpha_z).
\]
The empty product is understood to be $1$.
For $0\le u\le s-1$,
\[
\deg\!\bigl(x^u\Lambda_{i,H}(x)\bigr)
\le(s-1)+(n-d-1)=n-k-1.
\]
Thus, for every scalar $e\in K$, the evaluation vector of
$e\,x^u\Lambda_{i,H}(x)$ with the dual multipliers $w_j$ lies in
$C^\perp$.

Set $t_H=M/L_H$ and choose an $F_H$-basis
$e_1,\ldots,e_{t_H}$ of $S$. For a codeword
$c=(c_1,\ldots,c_n)\in C$, define the values downloaded
from helper $h\in H$ by
\[
D_{h,m}=\Tr_{K/F_H}\!\left(
c_h e_m w_h\Lambda_{i,H}(\alpha_h)\right),
\qquad 1\le m\le t_H.
\]
The parity checks from $e_mx^u\Lambda_{i,H}(x)$ give,
for $1\le m\le t_H$ and $0\le u\le s-1$,
\[
c_i e_m w_i\alpha_i^u\Lambda_{i,H}(\alpha_i)
=
-\sum_{h\in H}
c_h e_m w_h\alpha_h^u\Lambda_{i,H}(\alpha_h),
\]
because $\Lambda_{i,H}(\alpha_z)=0$ for every nonhelper $z\in Z$.
Taking $\Tr_{K/F_H}$ and using $\alpha_h\in F_H$ yields
\[
\Tr_{K/F_H}\!\left(
c_i e_m w_i\alpha_i^u\Lambda_{i,H}(\alpha_i)
\right)
=-\sum_{h\in H}\alpha_h^u D_{h,m}.
\]

The direct-sum decomposition of $K$ shows that
\[
\{e_m\alpha_i^u:1\le m\le t_H,\ 0\le u\le s-1\}
\]
is an $F_H$-basis of $K$.  Since
$w_i\Lambda_{i,H}(\alpha_i)\ne0$, multiplying this basis by that scalar
produces another $F_H$-basis of $K$.  The trace pairing
\[
K\times K\longrightarrow F_H,
\qquad
(x,y)\longmapsto\Tr_{K/F_H}(xy),
\]
is nondegenerate.  Hence the $st_H=[K:F_H]$ trace values on the left uniquely determine
$c_i$: if all of them vanish, $c_i$ is orthogonal under the trace
pairing to an $F_H$-basis of $K$, and therefore $c_i=0$.

Thus the $t_H$ values $D_{h,m}$ from each helper suffice
for all $u=0,\ldots,s-1$. Each $D_{h,m}$ is one symbol
over $F_H$, hence represents
$[F_H:F]=L_H$ symbols over $F$.  Moreover, after fixing an $F$-basis
of $F_H$, its coordinate maps are $F$-linear functionals of the helper
symbol $c_h$.  Thus each helper sends
\[
t_H[F_H:F]
=\frac{M}{L_H}L_H
=M
=\frac{\ell}{s}
\]
symbols over $F$.  The total repair bandwidth is therefore
$dM=d\ell/s$, exactly the cut-set bound. Since $(i,H)$ was
arbitrary, the code admits optimal $F$-linear repair in every
required repair configuration.
\end{proof}

\begin{remark}
The values $\Lambda_{i,H}(\alpha_h)$ need not lie in $F_H$, because only
the factors $\alpha_h^u$ are moved through the trace.  Moreover, the
repair subspace may depend on $(i,H)$: the model requires an optimal
scheme for every repair scenario, not a common repair space across all
scenarios.
\end{remark}

\subsection{Exact Minimum Subpacketization for
\texorpdfstring{$d(d+2)\geq 4(n+1)$}{d(d+2) >= 4(n+1)}}

The construction above and the lower bound from
\cref{thm:quadratic-global-lower} together yield the following exact formula.

\begin{theorem}
\label{thm:minimum-exact}
Let $3\le k<d\le n-1$, put $s=d+1-k$, and let $q$ be any prime
power.  If
\[
 d(d+2)\geq 4(n+1),
\]
then
\[
 \ellmin(n,k,d;q)=s\Delta_s(n,d)=sG_s(n).
\]
\end{theorem}

\begin{proof}
By \cref{thm:quadratic-global-lower,thm:degree-construction} and
\eqref{eq:Delta-G-basic},
\[
 sG_s(n)
 \le \ellmin(n,k,d;q)
 \le s\Delta_s(n,d)
 \le sG_s(n).
\]
Hence equality holds throughout.
\end{proof}

In their survey, Ramkumar \emph{et al.}
\cite[Open Problem~20]{RamkumarEtAl2022} asked for the minimum
subpacketization of an RS code that attains the cut-set bound in the
full-helper case $d=n-1$. For exact linear repair in the range
$3\le k<n-1$, the preceding theorem resolves the corresponding case
of Open Problem~20. The resulting formula is recorded below.

\begin{corollary}
\label{cor:full-helper-exact}
Let $d=n-1$ and $s=n-k$, where $3\le k<n-1$.  Then, for every prime
power $q$,
\[
 \ellmin(n,k,n-1;q)
 =s\Delta_s(n,n-1)
 =sG_s(n).
\]
\end{corollary}

\begin{proof}
Here $n\ge5$, and
\[
 (n-1)(n+1)-4(n+1)=(n+1)(n-5)\ge0.
\]
Hence the quadratic condition holds with $d=n-1$, and
\cref{thm:minimum-exact} applies.
\end{proof}
\subsection{Exact Formulas for \texorpdfstring{$G_s(m)$}{G\_s(m)}}
\label{sec:Gs}
By \cref{thm:minimum-exact}, the minimum subpacketization
equals $sG_s(n)$ whenever $d(d+2)\geq 4(n+1)$.
We now determine the exact threshold for the identity
$G_s(m)=A_s\primorial{m}$ and derive a finite formula
valid for all $m$.

We first introduce some notation.  For $s\ge3$, let
\[
 u_s=\max\{p<s:p\text{ is prime}\},\qquad
 v_s=\max\{p^a<s:p\text{ is prime},\ a\ge1\}.
\]
Thus $u_s$ and $v_s$ are, respectively, the largest prime and the
largest prime power below $s$.  Set $u_2=v_2=1$ and define
\[
 A_s=\lcmop(1,\ldots,s-1),\qquad
 \sigma_s=\pi(u_sv_s).
\]
For each prime $p<s$, put
\[
a_p=p^{\lceil\log_p s\rceil},
\]
and, for every nonempty set $B$ of primes below $s$, let
$\mu_s(B)$ denote the smallest integer at least $s$ whose
prime-divisor set is exactly $B$.
We determine the exact point at which $G_s(m)$ is attained by taking
the smallest prime power at least $s$ for each prime below $s$, and
then adjoining successive primes at least $s$.

\begin{theorem}
\label{thm:G-stabilization}
For all integers $s\ge2$ and $m\ge0$,
\[
 G_s(m)=A_s\primorial{m}
 \quad\Longleftrightarrow\quad
 m\ge\sigma_s.
\]
In particular, $m\ge\lfloor s^2/4\rfloor$ is sufficient.
\end{theorem}

\begin{proof}
For $s=2$, the product of any $m$ pairwise coprime integers
at least $2$ is at least $\primorial{m}$, and equality is
attained by the first $m$ primes.  Since $A_2=1$ and
$\sigma_2=0$, the assertion follows.  Henceforth assume $s\ge3$,
and recall that $c_s=\pi(s-1)$.

We first record two consequences of the definitions.  Bertrand's
postulate \cite{HardyWright2008} implies $2u_s>s-1$, and hence
\[
u_s<s\le2u_s\le u_s^2,
  \qquad u_s\le v_s<2u_s.
\]
In particular, $u_sv_s\ge s$, and $u_sv_s$ is composite.
By definition, $a_p$ is the smallest power of $p$ not less than $s$;
accordingly, $a_p/p$ is the largest power of $p$ strictly
below $s$.  It follows that
\[
 \frac{a_p}{p}\le v_s,\qquad
 a_p\le pv_s\le u_sv_s,\qquad\text{ and } a_{u_s}=u_s^2.
\]
Since $A_s$ is the product of the largest powers below $s$
of the distinct primes $p<s$, we also have
\begin{equation}
 A_s\primorial{c_s}=\prod_{p<s}a_p.
 \label{eq:all-prime-powers}
\end{equation}

For $m\ge c_s$, the $m$ integers
\[
 \{a_p:p<s\}\cup\{p_{c_s+1},\ldots,p_m\}
\]
are at least $s$ and pairwise coprime.  Their product is
$A_s\primorial{m}$ by \eqref{eq:all-prime-powers}, so
\[
G_s(m)\le A_s\primorial{m}\qquad(m\ge c_s).
\]
Conversely, if $0\le m<c_s$, the first $m$ numbers
$a_{p_1},\ldots,a_{p_m}$ form an admissible family and give
\[
 G_s(m)\le\prod_{j=1}^m a_{p_j}
   =\primorial{m}\prod_{j=1}^m\frac{a_{p_j}}{p_j}<\primorial{m}\prod_{p<s}\frac{a_p}{p}
   =A_s\primorial{m}.
\]
The inequality is strict because at least one prime $p<s$
is omitted and every $a_p/p>1$.  Thus equality in the
theorem is possible only for $m\ge c_s$.

We next bound the cost of separating a nonempty set $B$ of
small prime divisors. We claim that
\begin{equation}
 \prod_{p\in B}a_p
 \le\mu_s(B)(u_sv_s)^{|B|-1}.
 \label{eq:sharp-block-exchange}
\end{equation}
When $B=\{p\}$, we have $\mu_s(B)=a_p$, so equality holds.
Suppose next that $B=\{p,q\}$, where $p\ne q$.  Since
$pq\mid\mu_s(B)$,
\[
\frac{a_pa_q}{\mu_s(B)}
 \le\frac{a_p}{p}\frac{a_q}{q}.
\]
If either $a_p/p\le u_s$ or $a_q/q\le u_s$, the right-hand
side is at most $u_sv_s$, as required.  Otherwise,
$a_p/p>u_s\ge p$ and $a_q/q>u_s\ge q$.  Since these
quotients are powers of $p$ and $q$, respectively, we obtain
\[
 p^2\le\frac{a_p}{p}\le v_s,
 \qquad q^2\le\frac{a_q}{q}\le v_s.
\]
Therefore $pq\le v_s<s$.  As $\mu_s(B)$ is a multiple
of $pq$ and is at least $s$, necessarily
$\mu_s(B)\ge2pq$.  Consequently,
\[
 \frac{a_pa_q}{\mu_s(B)}
 \le\frac12\frac{a_p}{p}\frac{a_q}{q}
 \le\frac{v_s^2}{2}<u_sv_s,
\]
where the last inequality follows from $v_s<2u_s$.
Finally, if $|B|\ge3$, then $\mu_s(B)\ge\prod_{p\in B}p$,
so
\[
 \frac{\prod_{p\in B}a_p}{\mu_s(B)}
 \le\prod_{p\in B}\frac{a_p}{p}
 \le v_s^{|B|}
 \le(u_sv_s)^{|B|-1}.
\]
For the last inequality we used
$v_s<2u_s\le u_s^2\le u_s^{|B|-1}$.
This proves \eqref{eq:sharp-block-exchange} for every nonempty $B$.

We first show that the closed form holds for
$m\ge\sigma_s=\pi(u_sv_s)$.  Since
$u_sv_s\ge s$, this implies $m\ge c_s$.
Choose pairwise coprime integers $x_1,\ldots,x_m\ge s$
whose product is $G_s(m)$.  If some $x_j$ has a prime
divisor $p\ge s$, then $x_j=p$: otherwise replacing $x_j$
by $p$ would preserve both admissibility conditions and
strictly reduce the product.  Similarly, if every prime divisor of $x_j$ is below $s$, then $x_j$ must be the smallest integer at least $s$ with exactly the same set of prime divisors. Otherwise, replacing $x_j$ by that integer would preserve admissibility and strictly reduce the product.

Thus, after reordering, exactly $r$ of the $x_j$ have the
form $\mu_s(B_1),\ldots,\mu_s(B_r)$, with nonempty pairwise
disjoint sets $B_1,\ldots,B_r$ of primes below $s$;
the remaining $m-r$ integers are distinct primes at least
$s$.  Necessarily $0\le r\le c_s$.  Apply
\eqref{eq:sharp-block-exchange} to each $B_j$ and use
$a_p\le u_sv_s$ for every prime $p<s$ outside their union.
By \eqref{eq:all-prime-powers},
\[
\begin{aligned}
A_s\primorial{c_s}
=\prod_{p<s}a_p
&\le\left(\prod_{j=1}^r\mu_s(B_j)\right)
(u_sv_s)^{\sum_{j=1}^r(|B_j|-1)+c_s-\sum_j|B_j|}\\
&=(u_sv_s)^{c_s-r}\prod_{j=1}^r\mu_s(B_j).
\end{aligned}
\]
The same estimate holds for $r=0$, with empty products
interpreted as $1$.

Since the other $m-r$ factors are distinct primes at least
$s$, their product is bounded below by the product of the
first $m-r$ primes at least $s$, namely
\[
 \prod_{j=c_s+1}^{c_s+m-r}p_j
 =\frac{\primorial{m+c_s-r}}{\primorial{c_s}}.
\]
Multiplying the two estimates yields
\begin{equation}
 G_s(m)\ge\frac{A_s\primorial{c_s}}{(u_sv_s)^{c_s-r}}
       \frac{\primorial{m+c_s-r}}{\primorial{c_s}}=A_s\primorial{m}
       \prod_{j=1}^{c_s-r}\frac{p_{m+j}}{u_sv_s}.
 \label{eq:G-threshold-lower}
\end{equation}
Because $u_sv_s$ is composite and
$m\ge\pi(u_sv_s)$, we have $p_{m+1}>u_sv_s$.
Every factor in the final product is therefore at least
$1$ (and the product is $1$ if $r=c_s$).
Hence $G_s(m)\ge A_s\primorial{m}$, and equality follows
from the explicit construction above.

Conversely, suppose $G_s(m)=A_s\primorial{m}$.  The argument before
\eqref{eq:sharp-block-exchange} shows that $m\ge c_s$.
Assume, for a contradiction, that $m<\sigma_s$.  Since
$u_sv_s$ is composite, $m<\pi(u_sv_s)$ implies
$p_{m+1}<u_sv_s$; and $m\ge c_s$ implies $p_{m+1}\ge s$.
Thus
\[
s\le p_{m+1}<u_sv_s.
\]
Starting from the admissible family
\[
 \{a_p:p<s\}\cup\{p_{c_s+1},\ldots,p_m\},
\]
whose product is $A_s\primorial{m}$, we construct a
strictly smaller admissible family of the same size.

If $v_s=u_s$, replace the factor
$a_{u_s}=u_s^2=u_sv_s$ by $p_{m+1}$.  The new factor is
at least $s$ by the preceding inequality.
It is coprime to every remaining factor, since
$p_{m+1}$ does not occur in the original family.
The product decreases by the factor
$p_{m+1}/(u_sv_s)<1$.

If $v_s>u_s$, let $p<s$ be the prime for which
$v_s$ is a power of $p$.  Since $v_s$ is the largest
prime power below $s$, we have $a_p/p=v_s$.
Also $p\ne u_s$, as $a_{u_s}/u_s=u_s<v_s$.
Replace the two factors $a_p=pv_s$ and
$a_{u_s}=u_s^2$ by $pu_s$ and $p_{m+1}$.
Both new factors are at least $s$, since
\[
 pu_s\ge2u_s\ge s,
 \qquad p_{m+1}\ge s.
\]
They are relatively prime to one another and to the
unchanged factors: the primes $p$ and $u_s$ occur only in
the two removed factors, and $p_{m+1}>u_s$ is a new prime.
The ratio of the new product to the original product is
\[
 \frac{pu_s\,p_{m+1}}{(pv_s)u_s^2}
 =\frac{p_{m+1}}{u_sv_s}<1.
\]
Either replacement contradicts the optimality of
$A_s\primorial{m}$.  Hence $m\ge\sigma_s$.

It remains to verify the stated sufficient condition
$m\ge\lfloor s^2/4\rfloor$.  The case $s=2$ is immediate,
and for $3\le s\le10$ we have
$\sigma_s\le\lfloor s^2/4\rfloor$ by the values listed below.
For $s\ge11$, the prime bound
$p_j>j\log j$ for $j\ge2$ \cite[(3.12)]{RosserSchoenfeld1962}
gives
\[
 p_{\lfloor s^2/4\rfloor+1}
 >\frac{s^2}{4}\log\frac{s^2}{4}
 >(s-1)^2\ge u_sv_s.
\]
The middle inequality holds for $s\ge11$: the difference
between its two sides is positive at $s=11$ and has derivative
$s(\log(s/2)-3/2)+2>0$ for $s\ge11$.
Thus $\sigma_s=\pi(u_sv_s)\le\lfloor s^2/4\rfloor$,
as required.
\end{proof}

The threshold can be read off explicitly for small $s$:
\[
\begin{array}{c|rrrrrrrrr}
 s&2&3&4&5&6&7&8&9&10\\ \hline
 A_s&1&2&6&12&60&60&420&840&2520\\
 u_sv_s&1&4&9&12&25&25&49&56&63\\
 \sigma_s&0&2&4&5&9&9&15&16&18
\end{array}
\]
For instance, when $s=9$, the largest prime below $s$ is $7$ and the
largest prime power below $s$ is $8$, so the threshold is determined by
$7\cdot8=56$. In the family used above, the factors corresponding to
$2$ and $7$ are $16$ and $49$. Replacing them by $14$ and the next
unused prime decreases the product exactly while that prime is below
$56$, which illustrates why the product $u_sv_s$ appears in the exact
threshold.

Combining the threshold with \cref{thm:minimum-exact} yields the
following closed form whenever $d(d+2)\geq 4(n+1)$.

\begin{corollary}
\label{cor:closed-minimum}
Let $q$ be a prime power, $3\le k<d\le n-1$, and $s=d+1-k$. Assume $d(d+2)\geq 4(n+1)$. Then
\[
\ellmin(n,k,d;q)=sA_s\primorial{n}
\quad\Longleftrightarrow\quad
n\ge\sigma_s.
\]
In particular, for $n\ge5$,
\[
\ellmin(n,n-2,n-1;q)=2\primorial{n}.
\]
\end{corollary}

For fixed $s$, the closed form also determines the asymptotic behavior of
$G_s(m)$ directly.

\begin{corollary}
\label{cor:G-asymptotic}
For every fixed integer $s\ge2$, as $m\to\infty$,
\[
\log G_s(m)=(1+o(1))m\log m.
\]
\end{corollary}

\begin{proof}
For $m\ge\sigma_s$, use \cref{thm:G-stabilization} and
$\log\primorial{m}\sim m\log m$ from the prime number theorem
\cite{HardyWright2008}.
\end{proof}

The preceding theorem provides a closed form for $G_s(m)$
when $m\ge\sigma_s$. We now give a finite formula valid
also below this threshold. For each nonempty set $B$ of
primes below $s$, the integer $\mu_s(B)$ can be computed by
\[
\mu_s(B)=\min\left\{\prod_{p\in B}p^{e_p}:
  1\le e_p\le\lceil\log_p s\rceil,\quad
  \prod_{p\in B}p^{e_p}\ge s\right\},
\]
where the exponents are integers.  Indeed, an exponent exceeding
$\lceil\log_p s\rceil$ can be reduced to this value without
changing the set of prime divisors or making the product less
than $s$.

\begin{theorem}
\label{thm:G-finite-formula}
Let $s\ge2$ and $m\ge0$.  For each integer
$0\le r\le\min\{m,c_s\}$, let $B_1,\ldots,B_r$ range over
pairwise disjoint nonempty subsets of
$\{p_1,\ldots,p_{c_s}\}$.  Then
\begin{equation}
G_s(m)=
\min_{\substack{0\le r\le\min\{m,c_s\}\\ B_1,\ldots,B_r}}
  \left(\prod_{j=1}^{r}\mu_s(B_j)\right)
  \frac{\primorial{m+c_s-r}}{\primorial{c_s}}.
\label{eq:G-exact-formula}
\end{equation}
For $r=0$, the tuple is empty and its product is $1$.
If $m\ge\sigma_s$, a minimizing choice is $r=c_s$ and
$B_j=\{p_j\}$ for $1\le j\le c_s$.
\end{theorem}

\begin{proof}
Choose pairwise coprime integers $x_1,\ldots,x_m\ge s$
whose product is $G_s(m)$.  As in the proof of
\cref{thm:G-stabilization}, any $x_j$ with a prime divisor
at least $s$ must equal that prime: otherwise, replacing
$x_j$ by the prime would decrease the product and preserve
admissibility.  Likewise, if all prime divisors of $x_j$
are below $s$, minimality forces $x_j=\mu_s(B)$, where
$B$ is its set of prime divisors.

After reordering, suppose the latter type occurs $r$ times,
with prime-divisor sets $B_1,\ldots,B_r$.  These sets are
nonempty, pairwise disjoint, and contained in
$\{p_1,\ldots,p_{c_s}\}$; in particular,
$0\le r\le\min\{m,c_s\}$.  The remaining $m-r$ factors
are distinct primes at least $s$.  Their product is at least
the product of the first $m-r$ primes at least $s$, namely
\[
\prod_{j=c_s+1}^{c_s+m-r}p_j
=\frac{\primorial{m+c_s-r}}{\primorial{c_s}}.
\]
Consequently,
\[
G_s(m)\ge
\left(\prod_{j=1}^{r}\mu_s(B_j)\right)
\frac{\primorial{m+c_s-r}}{\primorial{c_s}},
\]
which is at least the minimum on the right-hand side of
\eqref{eq:G-exact-formula}.

Conversely, fix any $r$ and $B_1,\ldots,B_r$ allowed
in the theorem.  The $m$ integers
\[
\mu_s(B_1),\ldots,\mu_s(B_r),
\quad p_{c_s+1},\ldots,p_{c_s+m-r}
\]
are all at least $s$ and are pairwise coprime: the
$B_j$ are disjoint and contain only primes below $s$,
while the remaining factors are distinct primes at least
$s$.  Their product equals the corresponding expression
in \eqref{eq:G-exact-formula}, proving the reverse
inequality.

Finally, let $m\ge\sigma_s$ and choose $r=c_s$ and
$B_j=\{p_j\}$.  Since $\mu_s(\{p_j\})=a_{p_j}$,
\eqref{eq:all-prime-powers} gives
\[
\left(\prod_{j=1}^{c_s}\mu_s(\{p_j\})\right)
\frac{\primorial{m}}{\primorial{c_s}}
=A_s\primorial{m}=G_s(m),
\]
where the last equality follows from
\cref{thm:G-stabilization}.  Hence this choice attains
the minimum.
\end{proof}

We illustrate the formula for small $s$, including the exact
threshold values.
For $s=3$, $\mu_3(\{2\})=4$, and hence
\[
G_3(1)=3,\qquad G_3(2)=12=2\primorial{2}.
\]
For $s=4$, the possible nonempty sets of small primes
have
\[
\mu_4(\{2\})=4,\qquad
\mu_4(\{3\})=9,\qquad
\mu_4(\{2,3\})=6,
\]
which yields
\[
G_4(1)=4,\quad G_4(2)=20,\quad
G_4(3)=140,\quad G_4(4)=1260.
\]
For $s=5$, the corresponding values are
\[
\mu_5(\{2\})=8,\qquad
\mu_5(\{3\})=9,\qquad
\mu_5(\{2,3\})=6,
\]
and thus
\[
\begin{aligned}
&G_5(1)=5,\quad G_5(2)=30,\quad G_5(3)=210,\\
&G_5(4)=2310,\quad G_5(5)=27720.
\end{aligned}
\]
In particular, $G_5(2)=6\cdot5<5\cdot7$: combining the
small primes $2$ and $3$ in the factor $6$ is better than
using only primes at least $5$.  The values at
$m=\sigma_3=2$, $m=\sigma_4=4$, and $m=\sigma_5=5$
agree with the closed form in \cref{thm:G-stabilization}.

\section{Asymptotics for Fixed \texorpdfstring{$k$ and $d$}{k and d}}
\label{sec:fixed-parameters}

Fix a prime power $q$ and integers $3\le k<d$, put $s=d+1-k$,
and let $n\to\infty$ with $k$ and $d$ fixed.  We determine the order
of $\log\ellmin(n,k,d;q)$.  The lower bound comes from counting
distinct subset least common multiples, while the upper bound uses
cover-free families.  All logarithms in this section are natural.

Let $\tau(M)$ denote the number of positive divisors of
a positive integer $M$. For $x\ge1$, define
\begin{equation}
\mathfrak D(x)=\min\{M\in\mathbb Z: M\ge1,\ \tau(M)\ge x\}.
\label{eq:inverse-divisor}
\end{equation}

\begin{lemma}
\label{lem:subset-lcm-count}
Let $(a_1,\ldots,a_N)$ be a feasible $(N,t,s)$-tuple, where
$s\ge2$ and $1\le t<N$, and put $M=\lcmop(a_1,\ldots,a_N)$.
Then
\[
\tau(M)\ge\sum_{r=0}^{t}\binom Nr.
\]
Consequently,
\[
\Delta_s(N,t)\ge
\mathfrak D\!\left(\sum_{r=0}^{t}\binom Nr\right).
\]
\end{lemma}

\begin{proof}
For $A\subseteq[N]$, write $L_A=\lcmop(a_i:i\in A)$, with
$L_\varnothing=1$.  We claim that for $|A|,|B|\le t$,
\[
L_A\mid L_B\quad\Longrightarrow\quad A\subseteq B.
\]
Otherwise, take $i\in A\setminus B$.  Since $|B|\le t<N$,
we may extend $B$ to a $t$-subset
$J\subseteq[N]\setminus\{i\}$.  Then
$a_i\mid L_A\mid L_B\mid L_J$, and hence
\[
\frac{L_{J\cup\{i\}}}{L_J}=1<s,
\]
contrary to feasibility.  In particular, $L_A=L_B$ implies
$A=B$ whenever $|A|,|B|\le t$.  These
$\sum_{r=0}^{t}\binom Nr$ distinct values are all divisors of $M$,
which proves the first assertion.  The second follows from
\eqref{eq:inverse-divisor} and the definition of $\Delta_s(N,t)$.
\end{proof}

To use this counting bound, we need the asymptotic growth of
$\mathfrak D(x)$.  The following lemma is a direct consequence of
Wigert's theorem on the maximal order of the divisor function.

\begin{lemma}
\label{lem:inverse-divisor-estimate}
As $x\to\infty$,
\begin{equation}
\log\mathfrak D(x)
=\left(\frac{1}{\log2}+o(1)\right)\log x\log\log x.
\label{eq:inverse-divisor-asymptotic}
\end{equation}
\end{lemma}

\begin{proof}
For $Y\ge1$, let
\[
H(Y)=\max_{\substack{M\in\mathbb Z\\1\le M\le Y}}\log\tau(M).
\]
Wigert's theorem \cite{Wigert1907}, in the formulation given in
\cite{Elsholtz2019}, states that
\[
H(Y)=(\log2+o(1))\frac{\log Y}{\log\log Y}
\qquad(Y\to\infty).
\]
For $x\ge1$ and $Y\ge1$, the definitions yield the equivalence
\[
H(Y)\ge\log x
\quad\Longleftrightarrow\quad
\mathfrak D(x)\le Y.
\]
Fix $\varepsilon\in(0,1)$ and, for $x$ sufficiently large, set
\[
Y_\pm=
\exp\!\left(\frac{1\pm\varepsilon}{\log2}
\log x\log\log x\right).
\]
Since
\[
\log\log Y_\pm
=\log\log x+\log\log\log x+O_\varepsilon(1)
=(1+o(1))\log\log x,
\]
Wigert's theorem yields
\[
H(Y_\pm)=(1\pm\varepsilon+o(1))\log x.
\]
Thus, for all sufficiently large $x$,
\[
H(Y_-)<\log x<H(Y_+),
\]
and the preceding equivalence gives
\[
Y_-<\mathfrak D(x)\le Y_+.
\]
Taking logarithms and then letting $\varepsilon\to0$
proves \eqref{eq:inverse-divisor-asymptotic}.
\end{proof}

For the upper bound, recall that a family
$\mathcal A=\{A_1,\ldots,A_N\}\subseteq2^{[m]}$ is
\emph{$t$-cover-free} if
\[
A_i\not\subseteq\bigcup_{j\in J}A_j
\]
for every $i\in[N]$ and every $t$-subset
$J\subseteq[N]\setminus\{i\}$.
Let $\mcff(N,t)$ be the minimum possible size $m$ of the
ground set.

\begin{lemma}
\label{lem:CFF-upper}
For integers $s\ge2$ and $1\le t<N$,
\begin{equation}
\Delta_s(N,t)\le G_s(\mcff(N,t)).
\label{eq:Delta-CFF-upper}
\end{equation}
Consequently, for every prime power $q$ and
$3\le k<d\le n-1$, with $s=d+1-k$,
\begin{equation}
\ellmin(n,k,d;q)\le sG_s(\mcff(n,d)).
\label{eq:ell-CFF-upper}
\end{equation}
\end{lemma}

\begin{proof}
Choose a $t$-cover-free family
$\mathcal A=\{A_1,\ldots,A_N\}$ on $[m]$ and pairwise
coprime integers $g_1,\ldots,g_m\ge s$ such that
$\prod_{r=1}^{m}g_r=G_s(m)$.  Set
\[
a_i=\prod_{r\in A_i}g_r,
\qquad i\in[N].
\]
For any $i\notin J$ with $|J|=t$, pairwise coprimality implies
\[
\frac{\lcmop(a_j:j\in J\cup\{i\})}
     {\lcmop(a_j:j\in J)}
=
\prod_{r\in A_i\setminus\bigcup_{j\in J}A_j}g_r\ge s,
\]
since the set in the last product is nonempty by the
cover-free property.  Thus $(a_1,\ldots,a_N)$ is feasible.
Moreover,
\[
\lcmop(a_1,\ldots,a_N)
=\prod_{r\in\bigcup_{i=1}^{N}A_i}g_r
\le G_s(m).
\]
Taking $m=\mcff(N,t)$ proves \eqref{eq:Delta-CFF-upper}.
The bound \eqref{eq:ell-CFF-upper} follows from
\cref{thm:degree-construction}.
\end{proof}

Cover-free families are the set-system counterpart of disjunct
matrices in nonadaptive group testing and are closely related to
superimposed codes \cite{KautzSingleton1964}.
Porat and Rothschild \cite{PoratRothschild2011} gave explicit
$t$-disjunct constructions with $O(t^2\log N)$ rows.  In
particular, for fixed $t$,
\begin{equation}
\mcff(N,t)=O_t(\log N).
\label{eq:CFF-fixed-t}
\end{equation}

\begin{theorem}
\label{thm:fixed-kd-asymptotic}
Fix a prime power $q$ and integers $3\le k<d$.  Put
$s=d+1-k$ and let $n\to\infty$.
There is a constant $C_d>0$ depending only on $d$ such that
\begin{equation}
\left(\frac{d-3}{\log2}+o(1)\right)\log n\log\log n
\le\log\ellmin(n,k,d;q)\le(C_d+o(1))\log n\log\log n.
\label{eq:fixed-kd-two-sided}
\end{equation}
In particular,
\[
\log\ellmin(n,k,d;q)
=\Theta_{k,d}(\log n\log\log n).
\]
\end{theorem}

\begin{proof}
For the lower bound, consider any RS code satisfying the
repair requirement.  Fix three distinct anchors $a,b,c\in\Omega$,
write the remaining evaluation points as
$\omega_1,\ldots,\omega_{n-3}$, and let $\phi$ be the
projective normalization from \cref{lem:three-anchor}.  Put
\[
x_i=\phi(\omega_i),\qquad e_i=[F(x_i):F],
\qquad t=d-3.
\]
For any $i\in[n-3]$ and any $t$-subset
$J\subseteq[n-3]\setminus\{i\}$, choose the helper set
\[
H=\{a,b,c\}\cup\{\omega_j:j\in J\}.
\]
By \cref{lem:three-anchor} and the finite-field degree formula,
\[
\crf_F(H)=F(x_j:j\in J),\qquad
\crf_F(H\cup\{\omega_i\})=F(x_i,x_j:j\in J).
\]
Hence \cref{prop:global-crossratio} gives
\[
\frac{\lcmop(e_j:j\in J\cup\{i\})}
     {\lcmop(e_j:j\in J)}
=[\crf_F(H\cup\{\omega_i\}):\crf_F(H)]\ge s.
\]
Thus $(e_1,\ldots,e_{n-3})$ is a feasible
$(n-3,d-3,s)$-tuple.  The three-anchor identity also gives
\[
M:=\lcmop(e_1,\ldots,e_{n-3})
=[\crf_F(\Omega):F]\mid\frac{\ell}{s},
\]
where the divisibility is from \cref{prop:global-crossratio}.
Applying \cref{lem:subset-lcm-count} yields
\[
\ell\ge sM\ge s\mathfrak D(X_n),
\qquad
X_n:=\sum_{r=0}^{d-3}\binom{n-3}{r}.
\]
Since $d$ is fixed,
\[
X_n=\left(\frac{1}{(d-3)!}+o(1)\right)n^{d-3},
\]
so that
\[
\log X_n=(d-3)\log n+O(1),\qquad
\log\log X_n=(1+o(1))\log\log n.
\]
By \cref{lem:inverse-divisor-estimate},
\[
\log\mathfrak D(X_n)
=\left(\frac{d-3}{\log2}+o(1)\right)
\log n\log\log n.
\]
This proves the lower bound in \eqref{eq:fixed-kd-two-sided}.

For the upper bound, \eqref{eq:CFF-fixed-t} ensures the existence of a constant
$C_d>0$ for which $\mcff(n,d)\le C_d\log n$ for all
sufficiently large $n$.  Set $m_n=\lceil C_d\log n\rceil$.
The function $G_s(m)$ is increasing: deleting one factor from
a minimizing $(m+1)$-tuple shows that
$G_s(m+1)\ge sG_s(m)$.  Therefore, by
\cref{lem:CFF-upper},
\[
\ellmin(n,k,d;q)
\le sG_s(\mcff(n,d))\le sG_s(m_n).
\]
Since $s$ is fixed and $m_n=(C_d+o(1))\log n$,
\cref{cor:G-asymptotic} yields
\[
\log\ellmin(n,k,d;q)
\le(C_d+o(1))\log n\log\log n.
\]
This proves the upper bound.
\end{proof}

The order of $\log\ellmin(n,k,d;q)$ is thus determined for
fixed $k$ and $d$; the optimal leading constant remains open.

\section{Conclusion}
\label{sec:conclusion}

In this paper, we studied the minimum subpacketization of Reed-Solomon
codes admitting cut-set-optimal exact $F$-linear repair of every
coordinate from any $d$ helper nodes.  By translating the repair
requirements into degree constraints on cross-ratio fields and
combining the resulting lower bounds with a basis-transformation
construction, we established
\[
\ellmin(n,k,d;q)=s\Delta_s(n,d)=sG_s(n)
\]
whenever $d(d+2)\geq 4(n+1)$.  In particular, this determines the
minimum subpacketization in the full-helper case $d=n-1$ with
$3\le k<n-1$, resolving the corresponding exact linear-repair case of
Open Problem~20 of Ramkumar \emph{et al.}
\cite[Open Problem~20]{RamkumarEtAl2022}.  We also obtained explicit
formulas for $G_s(n)$ and determined the asymptotic order of
$\log\ellmin(n,k,d;q)$ for fixed $k$ and $d$.

The expression $sG_s(n)$, however, is not valid for all parameters.
The affine-plane construction in Appendix~\ref{app:affine-obstruction}
provides admissible parameters for which
\[
\ellmin(n,k,d;q)\le s\Delta_s(n,d)<sG_s(n).
\]
Thus, although $sG_s(n)$ is exact in the range above,
$s\Delta_s(n,d)$ is the natural candidate for a general formula.
This motivates the following conjecture.

\begin{conjecture}
For every prime power $q$ and all integers $3\le k<d\le n-1$,
with $s=d+1-k$,
\[
\ellmin(n,k,d;q)=s\Delta_s(n,d).
\]
\end{conjecture}

Our main result proves this conjecture whenever
$d(d+2)\geq 4(n+1)$.  Since the matching upper bound holds for all
admissible parameters, it remains to establish
$\ellmin(n,k,d;q)\ge s\Delta_s(n,d)$ outside this range.

\section*{Acknowledgment}
We used OpenAI's ChatGPT (GPT-6) to assist with polishing the exposition throughout the manuscript and checking some of the main proofs. The authors take full responsibility for the correctness, originality, and proper attribution of all content in this manuscript.

\appendices
\section{Proof of Lemma~\ref{lem:elementary-prime-product}}
\label{app:degree-arithmetic}

\begin{proof}
We first establish
\begin{equation}
\label{eq:auxiliary-prime-products}
p_{4t^2}<\primorial{t+1},
\qquad
p_{4t^2}^{\,2}<\primorial{2t}
\qquad(t\ge3).
\end{equation}
For $t=3$, these inequalities follow from
$p_{36}=151<210=\primorial{4}$ and
$151^2=22801<30030=\primorial{6}$.
It remains to consider $t\ge4$. All logarithms in this proof are natural.

For an even integer $x=2m\ge4$, put $B_x=\binom{2m}{m}$.
The central binomial coefficient exceeds the average of all
$x+1$ binomial coefficients, so $B_x>2^x/(x+1)$.
For each prime $p$, the factorial valuation formula yields
\[
v_p(B_x)=\sum_{j\ge1}
\left(\left\lfloor\frac{x}{p^j}\right\rfloor
      -2\left\lfloor\frac{x/2}{p^j}\right\rfloor\right).
\]
Each summand is $0$ or $1$ and vanishes when $p^j>x$.
Thus $p^{v_p(B_x)}\le x$ for every prime $p$, and
$B_x\le x^{\pi(x)}$. Taking logarithms and using
$\log2>2/3$ and $x+1<x^2$, we obtain
\begin{equation}
\label{eq:elementary-pi-bound}
\pi(x)>
\frac{x\log2-\log(x+1)}{\log x}
>\frac{2x}{3\log x}-2.
\end{equation}

Set $x_t=(t+2)!$ and $F_t=2x_t/(3\log x_t)$.
Since $\log720<7$, we have $F_4>480/7>66$.
For $t\ge4$, the inequality $x_t\ge t+3$ implies
\[
\frac{F_{t+1}}{F_t}
=\frac{t+3}{1+\log(t+3)/\log x_t}
\ge\frac{t+3}{2}\ge\frac72.
\]
Moreover,
\[
\frac72(4t^2+2)-\bigl(4(t+1)^2+2\bigr)
=10t^2-8t+1>0.
\]
Induction therefore proves $F_t>4t^2+2$ for every $t\ge4$.
Applying \eqref{eq:elementary-pi-bound} to $x_t$ yields
$\pi(x_t)>4t^2$, and hence
\[
p_{4t^2}<x_t=(t+2)!\le\primorial{t+1},
\]
where the last inequality follows from $p_j\ge j+1$.

To prove the second inequality in
\eqref{eq:auxiliary-prime-products}, put
\[
R_t=\frac{\primorial{2t}}{(\primorial{t+1})^2}.
\]
We have $R_4=4199/2310>1$, and, for $t\ge4$,
\[
\frac{R_{t+1}}{R_t}
=\frac{p_{2t+1}p_{2t+2}}{p_{t+2}^2}>1.
\]
Thus $R_t\ge R_4>1$, so
\[
p_{4t^2}^{\,2}<(\primorial{t+1})^2<\primorial{2t}.
\]
Together with the case $t=3$, this proves
\eqref{eq:auxiliary-prime-products}.

Now let $s,n,t$ satisfy the hypotheses of the lemma and set
$x=p_{n+c_s}$. Since $n+c_s\le4t^2$, we have $x\le p_{4t^2}$.
By \eqref{eq:degree-completion},
\[
G_s(n)\le G_s(n-1)x
<G_s(n-1)\primorial{t+1},
\]
which proves \eqref{eq:completion-one}.
The sequence $(\primorial{m})^{1/m}$ is increasing, since
$p_{m+1}>(\primorial{m})^{1/m}$. Therefore, for $2\le b\le n$,
\[
\primorial{bt}\ge(\primorial{2t})^{b/2}
>p_{4t^2}^{\,b}\ge x^b.
\]
Combining this with \eqref{eq:degree-completion} proves
\eqref{eq:completion-many}.
\end{proof}

\section{Strict Inequality Between
\texorpdfstring{$\Delta_s(n,d)$}{Delta\_s(n,d)} and
\texorpdfstring{$G_s(n)$}{G\_s(n)}}
\label{app:affine-obstruction}

An affine-plane construction yields feasible degree tuples with
$\Delta_s(n,d)<G_s(n)$ for a family of parameters.
Together with \cref{thm:degree-construction}, this shows that
$\ellmin(n,k,d;q)=sG_s(n)$ is not valid for all admissible parameters.

\begin{proposition}
\label{prop:affine-plane-obstruction}
Let $T\ge2$ be a prime power, and let $n,d,s$ be integers with
\[
n\ge T^2+1,\qquad 1\le d\le T-1,\qquad s\ge2.
\]
Put $v=\max\{T^2,n-T\}$. Then
\begin{equation}
\label{eq:affine-obstruction-arithmetic}
\Delta_s(n,d)\le G_s(v)
\le s^{-(n-v)}G_s(n)<G_s(n).
\end{equation}
If additionally $4\le d\le T-1$ and $2\le s\le d-2$,
then, with $k=d+1-s$, for every prime power $q$,
\begin{equation}
\label{eq:affine-obstruction-coding}
\ellmin(n,k,d;q)\le s\Delta_s(n,d)<sG_s(n).
\end{equation}
Here $T$ need not be related to the repair base-field size $q$.
\end{proposition}

\begin{proof}
In the affine plane $\mathbb F_T^2$, consider the $T^2$
nonvertical lines
\[
L_{a,b}=\{(x,ax+b):x\in\mathbb F_T\},
\qquad a,b\in\mathbb F_T,
\]
and the $T$ vertical lines
\[
V_c=\{(c,y):y\in\mathbb F_T\},\qquad c\in\mathbb F_T.
\]
Each line has $T$ points, and two distinct lines meet in at
most one point.

Put $h=\min\{T,n-T^2\}$ and $z=n-T^2-h$.
Take all $T^2$ nonvertical lines, any $h$ vertical lines, and
$z$ singleton sets on distinct new points outside $\mathbb F_T^2$.
These sets form a family $\mathcal A$ of exactly
$T^2+h+z=n$ members. Since the nonvertical lines cover $\mathbb F_T^2$,
the ground set $\mathcal U=\bigcup_{A\in\mathcal A}A$ satisfies
\[
|\mathcal U|=T^2+z=n-h=\max\{T^2,n-T\}=v.
\]

For each line $A\in\mathcal A$, any $d$ other members cover
at most $d<T$ of its points. Each singleton contains a point
belonging to no other member. Hence, for every $A\in\mathcal A$
and every $\mathcal H\subseteq\mathcal A\setminus\{A\}$ with
$|\mathcal H|=d$,
\[
A\setminus\bigcup_{B\in\mathcal H}B\ne\varnothing.
\]

Choose pairwise coprime integers $\gamma_u\ge s$
($u\in\mathcal U$) with
$\prod_{u\in\mathcal U}\gamma_u=G_s(v)$, and define
\[
a_A=\prod_{u\in A}\gamma_u\qquad(A\in\mathcal A).
\]
Since the weights are pairwise coprime, the least common
multiple over any subfamily equals the product of the weights
on its union. Thus
\[
\frac{\lcmop(a_B:B\in\mathcal H\cup\{A\})}
     {\lcmop(a_B:B\in\mathcal H)}
=\prod_{u\in A\setminus\bigcup_{B\in\mathcal H}B}\gamma_u
\ge s.
\]
Therefore $(a_A)_{A\in\mathcal A}$ is feasible for $(n,d,s)$.
Because $\mathcal A$ covers $\mathcal U$,
\[
\Delta_s(n,d)\le\lcmop(a_A:A\in\mathcal A)
=\prod_{u\in\mathcal U}\gamma_u=G_s(v).
\]

In a minimizing tuple for $G_s(n)$, any $v$ entries have
product at least $G_s(v)$, while each remaining entry is
at least $s$. Consequently,
\[
G_s(n)\ge s^{n-v}G_s(v).
\]
Since $n-v=h\ge1$ and $s\ge2$, we have $s^{-(n-v)}<1$.
This proves \eqref{eq:affine-obstruction-arithmetic}.

Under the additional assumptions,
$3\le k=d+1-s\le d-1$ and $d\le T-1<n$.
The coding bound then follows from
\cref{thm:degree-construction}.
\end{proof}


\begin{thebibliography}{99}
\bibitem{Dimakis2010}
A.~G. Dimakis, P.~B. Godfrey, Y.~Wu, M.~J. Wainwright, and K.~Ramchandran,
``Network coding for distributed storage systems,''
\emph{IEEE Trans. Inf. Theory}, vol.~56, no.~9, pp.~4539--4551, Sep.~2010.

\bibitem{RashmiShahKumar2011}
K.~V. Rashmi, N.~B. Shah, and P.~V. Kumar,
``Optimal exact-regenerating codes for distributed storage at the MSR and MBR points via a product-matrix construction,''
\emph{IEEE Trans. Inf. Theory}, vol.~57, no.~8, pp.~5227--5239, Aug.~2011.

\bibitem{YeBarg2017}
M.~Ye and A.~Barg,
``Explicit constructions of high-rate MDS array codes with optimal repair bandwidth,''
\emph{IEEE Trans. Inf. Theory}, vol.~63, no.~4, pp.~2001--2014, Apr.~2017.

\bibitem{GoparajuTamoCalderbank2014}
S.~Goparaju, I.~Tamo, and R.~Calderbank,
``An improved sub-packetization bound for minimum storage regenerating codes,''
\emph{IEEE Trans. Inf. Theory}, vol.~60, no.~5, pp.~2770--2779, May~2014.

\bibitem{AlrabiahGuruswami2021}
O.~Alrabiah and V.~Guruswami,
``An exponential lower bound on the sub-packetization of minimum storage regenerating codes,''
\emph{IEEE Trans. Inf. Theory}, vol.~67, no.~12, pp.~8086--8093, Dec.~2021.

\bibitem{BalajiVajhaKumar2022}
S.~B. Balaji, M.~Vajha, and P.~Vijay Kumar,
``Lower bounds on the sub-packetization level of MSR codes and characterizing optimal-access MSR codes achieving the bound,''
\emph{IEEE Trans. Inf. Theory}, vol.~68, no.~10, pp.~6452--6471, Oct.~2022.

\bibitem{LiuLiTang2018}
Y.~Liu, J.~Li, and X.~Tang,
``Explicit constructions of high-rate MSR codes with optimal access property over small finite fields,''
\emph{IEEE Trans. Commun.}, vol.~66, no.~10, pp.~4405--4413, Oct.~2018.

\bibitem{LiTangTian2018}
J.~Li, X.~Tang, and C.~Tian,
``A generic transformation to enable optimal repair in MDS codes for distributed storage systems,''
\emph{IEEE Trans. Inf. Theory}, vol.~64, no.~9, pp.~6257--6267, Sep.~2018.

\bibitem{LiuLiTang2023}
Y.~Liu, J.~Li, and X.~Tang,
``A generic transformation to enable optimal repair/access MDS array codes with multiple repair degrees,''
\emph{IEEE Trans. Inf. Theory}, vol.~69, no.~7, pp.~4407--4428, Jul.~2023.

\bibitem{WangLiHuYe2023}
N.~Wang, G.~Li, S.~Hu, and M.~Ye,
``Constructing MSR codes with subpacketization $2^{n/3}$ for $k+1$ helper nodes,''
\emph{IEEE Trans. Inf. Theory}, vol.~69, no.~6, pp.~3775--3792, Jun.~2023.

\bibitem{LiWangHuYe2024}
G.~Li, N.~Wang, S.~Hu, and M.~Ye,
``MSR codes with linear field size and smallest sub-packetization for any number of helper nodes,''
\emph{IEEE Trans. Inf. Theory}, vol.~70, no.~11, pp.~7790--7806, Nov.~2024.

\bibitem{ReedSolomon1960}
I.~S. Reed and G.~Solomon,
``Polynomial codes over certain finite fields,''
\emph{J. Soc. Ind. Appl. Math.}, vol.~8, no.~2, pp.~300--304, Jun.~1960.

\bibitem{ShanmugamEtAl2014}
K.~Shanmugam, D.~S. Papailiopoulos, A.~G. Dimakis, and G.~Caire,
``A repair framework for scalar MDS codes,''
\emph{IEEE J. Sel. Areas Commun.}, vol.~32, no.~5, pp.~998--1007, May~2014.

\bibitem{GW2017}
V.~Guruswami and M.~Wootters,
``Repairing Reed-Solomon codes,''
\emph{IEEE Trans. Inf. Theory}, vol.~63, no.~9, pp.~5684--5698, Sep.~2017.

\bibitem{DauMilenkovic2017}
H.~Dau and O.~Milenkovic,
``Optimal repair schemes for some families of full-length Reed-Solomon codes,''
in \emph{Proc. IEEE Int. Symp. Inf. Theory (ISIT)}, Aachen, Germany,
Jun.~2017, pp.~346--350.

\bibitem{DauEtAl2021}
S.~H.~Dau, T.~X. Dinh, H.~M. Kiah, T.~T. Luong, and O.~Milenkovic,
``Repairing Reed-Solomon codes via subspace polynomials,''
\emph{IEEE Trans. Inf. Theory}, vol.~67, no.~10, pp.~6395--6407, Oct.~2021.

\bibitem{LiWangJafarkhani2019}
W.~Li, Z.~Wang, and H.~Jafarkhani,
``On the sub-packetization size and the repair bandwidth of Reed-Solomon codes,''
\emph{IEEE Trans. Inf. Theory}, vol.~65, no.~9, pp.~5484--5502, Sep.~2019.

\bibitem{TYB2017}
I.~Tamo, M.~Ye, and A.~Barg,
``Optimal repair of Reed-Solomon codes: Achieving the cut-set bound,''
in \emph{Proc. 58th Annu. IEEE Symp. Found. Comput. Sci. (FOCS)},
Berkeley, CA, USA, Oct.~2017, pp.~216--227.

\bibitem{TYB2019}
I.~Tamo, M.~Ye, and A.~Barg,
``The repair problem for Reed-Solomon codes: Optimal repair of single and multiple erasures with almost optimal node size,''
\emph{IEEE Trans. Inf. Theory}, vol.~65, no.~5, pp.~2673--2695, May~2019.

\bibitem{ChenYeBarg2020}
Z.~Chen, M.~Ye, and A.~Barg,
``Enabling optimal access and error correction for the repair of Reed-Solomon codes,''
\emph{IEEE Trans. Inf. Theory}, vol.~66, no.~12, pp.~7439--7456, Dec.~2020.

\bibitem{GuruswamiJiang2019}
V.~Guruswami and H.~Jiang,
``Near-optimal repair of Reed-Solomon codes with low subpacketization,''
in \emph{Proc. IEEE Int. Symp. Inf. Theory (ISIT)}, Paris, France, Jul.~2019, pp.~1077--1081.

\bibitem{Qiu2026}
J.~Qiu, W.~Fang, S.-T.~Xia, and F.-W.~Fu,
``Reed-Solomon codes with optimal repair bandwidth: A basis-transformation approach,''
\emph{IEEE Trans. Inf. Theory}, early access, 2026, DOI: 10.1109/TIT.2026.3745385.

\bibitem{RamkumarEtAl2022}
V.~Ramkumar, S.~B. Balaji, B.~Sasidharan, M.~Vajha,
M.~N. Krishnan, and P.~V. Kumar,
``Codes for distributed storage,''
\emph{Found. Trends Commun. Inf. Theory}, vol.~19, no.~4,
pp.~547--813, 2022.

\bibitem{LidlNiederreiter1997}
R.~Lidl and H.~Niederreiter,
\emph{Finite Fields}, 2nd~ed.
Cambridge, U.K.: Cambridge Univ. Press, 1997.

\bibitem{HardyWright2008}
G.~H. Hardy and E.~M. Wright,
\emph{An Introduction to the Theory of Numbers}, 6th~ed., revised by D.~R. Heath-Brown and J.~H. Silverman.
Oxford, U.K.: Oxford Univ. Press, 2008.

\bibitem{Wigert1907}
S.~Wigert,
``Sur l'ordre de grandeur du nombre des diviseurs d'un entier,''
\emph{Ark. Mat. Astron. Fys.}, vol.~3, no.~18, pp.~1--9, 1907.

\bibitem{Elsholtz2019}
C.~Elsholtz, M.~Technau, and N.~Technau,
``The maximal order of iterated multiplicative functions,''
\emph{Mathematika}, vol.~65, no.~4, pp.~990--1009, 2019.

\bibitem{KautzSingleton1964}
W.~H. Kautz and R.~C. Singleton,
``Nonrandom binary superimposed codes,''
\emph{IEEE Trans. Inf. Theory}, vol.~10, no.~4, pp.~363--377, Oct.~1964.

\bibitem{PoratRothschild2011}
E.~Porat and A.~Rothschild,
``Explicit nonadaptive combinatorial group testing schemes,''
\emph{IEEE Trans. Inf. Theory}, vol.~57, no.~12, pp.~7982--7989, Dec.~2011.

\bibitem{RosserSchoenfeld1962}
J.~B. Rosser and L.~Schoenfeld,
``Approximate formulas for some functions of prime numbers,''
\emph{Illinois J. Math.}, vol.~6, no.~1, pp.~64--94, Mar.~1962.
\end{thebibliography}
\end{document}